\documentclass[a4paper,UKenglish,cleveref, autoref, thm-restate]{lipics-v2021}

\usepackage{ifthen}
\newboolean{long}
\newcommand{\AppendixSymbol}{\ding{72}}
\setboolean{long}{true} %
\ifthenelse{\boolean{long}}{
	\NewDocumentEnvironment{prooflater}{m}{\begin{proof}}{\end{proof}\ignorespacesafterend}
	\NewDocumentEnvironment{proofsketch}{o +b}{}{\ignorespacesafterend}
	\newcommand{\restateref}[1]{}
	\NewDocumentEnvironment{statelater}{m}{}{}
}{
	\NewDocumentEnvironment{prooflater}{m +b}{%
		\expandafter\global\expandafter\def\csname#1\endcsname{\begin{proof}#2\end{proof}}%
	}{\ignorespacesafterend}
	\NewDocumentEnvironment{proofsketch}{O{Proof Sketch}}{\begin{proof}[#1]}{\end{proof}\ignorespacesafterend}
	\usepackage{apptools}
	\newcommand{\restateref}[1]{[\IfAppendix{\hyperref[#1]{\AppendixSymbol{}}}{\hyperref[#1*]{\AppendixSymbol{}}}]}

	\NewDocumentEnvironment{statelater}{m +b}{%
		\expandafter\global\expandafter\def\csname#1\endcsname{#2}%
	}{\ignorespacesafterend}
}

\pdfoutput=1 %
\hideLIPIcs  %

\graphicspath{{./graphics/}}%

\title{The Peculiarities of Extending Queue Layouts} %

\author{Thomas Depian}{Algorithms and Complexity Group, TU Wien, Vienna, Austria}{tdepian@ac.tuwien.ac.at}{https://orcid.org/0009-0003-7498-6271}{}

\author{Simon D.~Fink}{Algorithms and Complexity Group, TU Wien, Vienna, Austria}{sfink@ac.tuwien.ac.at}{https://orcid.org/0000-0002-2754-1195}{}

\author{Robert Ganian}{Algorithms and Complexity Group, TU Wien, Vienna, Austria}{rganian@ac.tuwien.ac.at}{https://orcid.org/0000-0002-7762-8045}{}

\author{Martin N\"ollenburg}{Algorithms and Complexity Group, TU Wien, Vienna, Austria}{noellenburg@ac.tuwien.ac.at}{https://orcid.org/0000-0003-0454-3937}{}

\authorrunning{T.~Depian, S.~D.~Fink, R.~Ganian, M.~N\"ollenburg} %

\Copyright{Thomas Depian, Simon D.~Fink, Robert Ganian, Martin N\"ollenburg} %

\ccsdesc[500]{Theory of computation~Fixed parameter tractability}
\ccsdesc[500]{Mathematics of computing~Graphs and surfaces}

\keywords{Queue layouts, Parameterized complexity, Linear layouts, Extension problems} %

\acknowledgements{All authors acknowledge support from the Vienna Science and Technology Fund (WWTF) [10.47379/ICT22029]. Robert Ganian and Thomas Depian furthermore acknowledge support from the Austrian Science Fund (FWF) [10.55776/Y1329].
	We want to thank 
	P.\ Bachmann, %
	T.\ Brand, %
	G.\ Gutowski, %
	M.\ Pfretzschner, %
	I.\ Rutter, %
	P.\ Stumpf, %
	S.\ Wolf, and %
	J.\ Zink %
	who at the HOMONOLO 2024 workshop helped us uncover and independently prove the relationship to permutation graphs used in \Cref{sec:only-missing edges}.
	We also want to thank the organizers for making this enjoyable event possible.
}

\nolinenumbers %
\usepackage[T1]{fontenc}
\usepackage{graphicx}
\graphicspath{{./graphics/}}

\usepackage{pifont}
\usepackage{cite}
\usepackage{mdframed}
\usepackage{xspace}
\usepackage{complexity}
\usepackage{todonotes}
\usepackage{mathtools}
\usepackage{amsmath}
\usepackage{booktabs}
\usepackage{multirow}

\newtheorem{property}{Property}

\crefname{property}{property}{properties}
\Crefname{property}{Property}{Properties}

\newcommand{\Vadd}{\ensuremath{V_{\text{add}}}\xspace}
\newcommand{\nadd}{\ensuremath{n_{\text{add}}}\xspace}

\newcommand{\Eadd}{\ensuremath{E_{\text{add}}}\xspace}
\newcommand{\madd}{\ensuremath{m_{\text{add}}}\xspace}
\newcommand{\EaddH}{\ensuremath{E_{\text{add}}^H}\xspace}

\NewDocumentCommand{\ql}{o}{\ensuremath{\langle\prec\IfNoValueF{#1}{_{#1}},\sigma\IfNoValueF{#1}{_{#1}}\rangle}\xspace}
\newcommand{\instance}{\ensuremath{\mathcal{I}}\xspace}
\newcommand{\instanceLong}{\ensuremath{\left(\ell, G, H, \ql[H]\right)}\xspace}

\NewDocumentCommand{\nestingRelation}{o}{\ensuremath{\Cap\IfNoValueF{#1}{_{#1}}}\xspace}

\newcommand{\intervalPlacing}[1]{\ensuremath{\curlyvee(#1)}}

\newcommand{\probname}[1]{{\normalfont\textsc{#1}}}
\newcommand{\probdef}[3]{%
	\begin{mdframed}
		\probname{#1}
		\begin{description}
			\item[Given] #2
			\item[Question] #3
		\end{description}
	\end{mdframed}
}%
\newcommand{\MCC}{\probname{McC}\xspace}
\newcommand{\QLELong}{\probname{Queue Layout Extension}\xspace}
\newcommand{\QLE}{\probname{QLE}\xspace}

\newcommand{\SLE}{\probname{SLE}\xspace}

\newcommand{\Size}[1]{\ensuremath{\left\vert #1 \right\vert}}
\newcommand{\BigO}[1]{\ensuremath{\mathcal{O}(#1)}}

\let\oldrestatable\restatable
\def\restatable{\expandafter\oldrestatable}

\begin{document}
\maketitle              %
\begin{abstract}
We consider the problem of computing $\ell$-page queue layouts, which are linear arrangements of vertices accompanied with an assignment of the edges to pages {from one to $\ell$} %
that avoid the nesting of edges on any of the pages. Inspired by previous work in the extension of stack layouts, here we consider the setting of extending a partial $\ell$-page queue layout into a complete one and primarily analyze the problem through the refined lens of parameterized complexity. We obtain novel algorithms and lower bounds which provide a detailed picture of the problem's complexity under various measures of incompleteness, and identify surprising distinctions between queue and stack layouts in the extension setting.
\end{abstract}

\section{Introduction}

The computation of fundamental representations or visualizations of graphs is a core research topic in several fields including theoretical computer science, computational geometry and graph drawing. While some representations can be computed in polynomial time (e.g., planar drawings of planar graphs), for many others computing a valid representation from the input graph alone is intractable (and, in particular, \NP-hard when phrased as a decision question). 

Nevertheless, it is entirely plausible---and often natural---that a fixed partial representation (such as a %
partial drawing) is already available, and our task is merely to complete this into a full representation. This perspective has recently led to the systematic study of graph representations through the perspective of so-called \emph{extension problems}. A classical result in this line of work is the %
{%
	linear}-time algorithm for extending partial planar drawings~\cite{ABF+.TPP.2015}; however, when the representations are \NP-hard to compute ``from scratch'' one cannot expect polynomial-time algorithms in the strictly more general extension setting. That is why work on extension problems has predominantly been carried out through the lens of parameterized complexity~\cite{EGH+.ENC.2020,EGH+.EP1.2020,GHK+.COE.2021,BDLMN.UBT.2023,BDLMN.UBT.2021,BGK+.EOP.2023,DFGN.PCE.2024,DFF+.PTG.2025}, 
where one can use parameters to capture how much of the input graph still needs to be added to the partial representation.

In this article we present the next step in the aforementioned systematic investigation of extension problems, specifically targeting the notion of \emph{queue layouts}. An $\ell$-page queue layout of a graph $G$ is a linear arrangement of $V(G)$ accompanied with an assignment of $E(G)$ to pages from $[\ell]$ such that the preimage of each page contains no pair of nesting edges (see \Cref{sec:preliminaries} and \Cref{fig:example}). Queue layouts have been extensively studied from a variety of contexts, including their relationship to planarity~\cite{BFG+.PGB.2019,DJM+.PGH.2020}, classical as well as parameterized complexity~\cite{HR.LOG.1992,BGMN.PAQ.2022,BGMN.PAQ.2020} and connections to other graph parameters~\cite{Wie.QNG.2017,DEH+.SNi.2022}. However, unlike for the closely related notion of stack layouts (which differ by excluding crossings instead of nestings)~\cite{CLR.EGB.1987,Bil.Egb.1992,DFGN.PCE.2024}, %
to the best of our knowledge up to now nothing was known about the extension problem for queue layouts:

\begin{figure}[t]
	\centering
	\includegraphics[page=1]{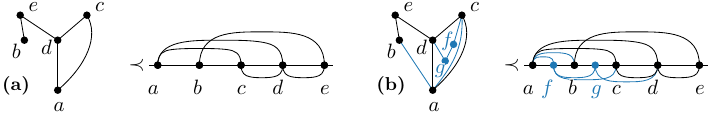}
	\caption{\textbf{(a)} Graph $H$ with a two-page queue layout $\ql[H]$.
	Observe the twist on the upper page {%
		between the edges $ac$ and $be$}. \textbf{(b)} An extension of $H$ and \ql[H] by new vertices and edges in blue.}
	\label{fig:example}
\end{figure}

\probdef{\QLELong (\QLE)}{An integer $\ell \geq 1$, a graph $G$, a subgraph $H$ of $G$, and an $\ell$-page queue layout \ql[H] of $H$.}{Does there exist an $\ell$-page queue layout \ql[G] of $G$ that extends \ql[H]?}

\QLE is easily seen to be \NP-hard already when $\ell=1$, as it generalizes the \NP-hard problem of computing a $1$-page queue layout from scratch~\cite{HR.LOG.1992}---indeed, this occurs when $V(H)=\emptyset$. An interesting special case occurs when $V(H)=V(G)$ and $E(H)=\emptyset$, i.e., %
{%
	$H$ contains all vertices of $G$ but none of its edges}: in this case, \QLE coincides with the well-studied problem of computing an $\ell$-page queue layout with a fixed linear arrangement of the vertices, which is known to be polynomial-time solvable~\cite{HR.LOG.1992}. Since the case where $E(H)=E(G)$ is trivial, this raises the question of whether \QLE is polynomial-time solvable whenever $V(G)=V(H)$. We begin our investigation by building upon %
{%
	a known relation between queue layouts with fixed spine order and colorings of %
specifically constructed 
permutation graphs~\cite{DW.LLG.2004,DM.POS.1941}}
to establish our first minor contribution:

\begin{restatable}\restateref{thm:only-edges-hard}{theorem}{thmVGisVH}
\label{thm:only-edges-hard}
    When restricted to the case where $V(G)=V(H)$, \QLE is \textup{(1)} \NP-complete and \textup{(2)} in \XP\ when parameterized by the number $\ell$ of pages.
\end{restatable}

\enlargethispage{1em}
By slightly adapting the ideas used for extending stack layouts~\cite[Theorem 3.2]{DFGN.PCE.2024}, we obtain our second minor contribution: fixed-parameter tractability for the special case above w.r.t.\ the number of missing edges.

\begin{restatable}\restateref{thm:only-edges-fpt}{theorem}{theoremOnlyEdgesFPT}
	\label{thm:only-edges-fpt}
	When restricted to the case where $V(G)=V(H)$, \QLE when only $\madd$ edges are missing from \ql[H] is fixed-parameter tractable when parameterized by $\madd$.%
\end{restatable}

Moving on to the general case where both edges and vertices may be missing\footnote{We say that an element, i.e., a vertex or edge, is \emph{missing} if it is in $G$ but not in $H$.} from $H$, we begin by providing our third and final ``minor'' contribution: 

\begin{restatable}\restateref{thm:kappa-xp}{theorem}{theoremKappaXP}
	\label{thm:kappa-xp}
    \QLE is in \XP\ when parameterized by the number $\kappa$ of missing elements.
\end{restatable}

Similarly to \Cref{thm:only-edges-fpt}, the proof of the above result follows by adapting the ideas from an analogous algorithm in the setting of stack layouts~\cite[Theorem 5.1]{DFGN.PCE.2024}. Essentially, while \Cref{thm:only-edges-hard,thm:only-edges-fpt,thm:kappa-xp} are important pieces in our overall understanding of \QLE, in terms of technical contributions they merely set the stage for our three main results. The first of these establishes a corresponding lower bound to Theorem~\ref{thm:kappa-xp}:

\begin{restatable}\restateref{thm:kappa-w1}{theorem}{theoremKappaWHard}
	\label{thm:kappa-w1}
	\QLE is \W\textup{[1]}-hard when parameterized by the number $\kappa$ of missing elements.
\end{restatable}

Theorem~\ref{thm:kappa-w1} also has a matching result in the stack layout setting~\cite[Theorem 6.4]{DFGN.PCE.2024}---however, the gadgets used to establish the previous lower bound seem impossible to directly translate to the queue layout setting, and circumventing this issue by redesigning the reduction required overcoming several technical difficulties that were not present in the stack layout setting. 

Since Theorem~\ref{thm:kappa-w1} excludes fixed-parameter algorithms w.r.t.\ the number of missing elements alone, the natural question is whether one can at least obtain such an algorithm if the parameter also includes the number $\ell$ of pages. This directly corresponds to the main open question highlighted in the previous work on stack layout extension~\cite[Section 8]{DFGN.PCE.2024}. As our second main contribution, for \QLE\ we answer this in the affirmative:

\begin{restatable}{theorem}{thmkappapagesfpt}
	\label{thm:kappa-pages-fpt}
	\QLE is fixed-parameter tractable w.r.t.\ the number $\ell$ of pages plus the number $\kappa$ of missing elements.%
\end{restatable}

The proof of Theorem~\ref{thm:kappa-pages-fpt} relies on a combination of branching and structural insights into queue layouts that %
allow us to reduce the problem to solving a set of \probname{2-Sat} instances. One high-level take-away from this result is that even if both vertices and edges are missing from $H$, \QLE\ seems to behave differently from stack layout extension---in particular, the approach used to establish \Cref{thm:kappa-pages-fpt} does not seem applicable to the latter. As our final result, we provide ``hard'' evidence for this intuition by establishing the polynomial-time tractability of the special case where only two vertices and %
their incident edges are missing:

\begin{restatable}\restateref{thm:two-missing-vertices-poly}{theorem}{theoremTwoMissingVerticesPoly}
    \label{thm:two-missing-vertices-poly}
    When restricted to the case where $H$ is obtained by deleting two vertices from $G$, \QLE can be solved in polynomial time.
\end{restatable}

The reason Theorem~\ref{thm:two-missing-vertices-poly} is noteworthy is that it directly contrasts the known \NP-hardness of the corresponding restriction in the stack layout setting~\cite[Theorem 4.2]{DFGN.PCE.2024}.
The proof %
is also non-trivial: the difficulty is that many missing edges may need to be assigned to 
a large number of available pages. We resolve this by developing a set of reduction rules based on novel insights into the behavior of queue layouts.
\Cref{tab:results} summarizes our results.

\begin{table}[t]
    \centering
    \caption{Overview of our results. \SLE is the problem of extending stack layouts %
    and VEDD = 2 indicates two missing vertices.
    Results for \QLE are shown in this paper.}
    \setlength{\tabcolsep}{.42em}
    \begin{tabular}{ccccc}
        \toprule
        Problem & $V(H) = V(G)$ & $\kappa$ & $\kappa + \ell$ & VEDD = 2\\  
        \midrule
        \multirow{2}{*}{\SLE} & \NP c, %
        \FPT(\madd)%
        & \XP, \W[1]-hard%
        & Open & \NP c%
        \\
        &{%
        	Thms.~1~\cite{Ung.kCC.1988} \&~3.2~\cite{DFGN.PCE.2024}}&{%
        	Thms.~5.1 \&~6.4~\cite{DFGN.PCE.2024}}&&{%
        	Thm.~4.2~\cite{DFGN.PCE.2024}}\\
        \midrule
        \multirow{2}{*}{\QLE} & \NP c, \FPT(\madd) & \XP, \W[1]-hard & \FPT & Polynomial time \\
        & Thms.\ \ref{thm:only-edges-hard} \& \ref{thm:only-edges-fpt} & Thms.\ \ref{thm:kappa-xp} \& \ref{thm:kappa-w1} & Thm.\ \ref{thm:kappa-pages-fpt} & Thm.\ \ref{thm:two-missing-vertices-poly}\\
        \bottomrule
    \end{tabular}
    \label{tab:results}
\end{table}

\section{Preliminaries}
\label{sec:preliminaries}
We assume familiarity with standard notions from graph theory~\cite{Die.GT4.2012} and parameterized complexity~\cite{CFK+.PA.2015}.
Unless stated otherwise, all graphs $G$ %
are simple and undirected with vertex set $V(G)$ and edge set $E(G)$.
We let $[\ell]$ denote the set $\{1, \ldots, \ell\}$ for $\ell \geq 1$ and $G[X]$ the subgraph of $G$ induced %
{%
	by} $X$ for $X \subseteq V(G)$.

Given two edges $u_1v_1$ and $u_2v_2$ and a linear order $\prec$ of $V(G)$ with $u_1 \prec u_2 \prec v_2 \prec v_1$, we say that $u_1v_1$ \emph{nests} $u_2v_2$ or, conversely, that $u_2v_2$ \emph{is nested} by $u_1v_1$.
In addition, the edges are in a \emph{nesting relation}, indicated as $u_1v_1 \nestingRelation[\prec_G] u_2v_2$.
We omit $\prec_G$ if it is clear from context and observe that only edges with pairwise distinct endpoints can be in a nesting relation.
An \emph{$\ell$-page queue layout} consists of a linear order $\prec_G$ of $V(G)$ (the \emph{spine order}) and a function $\sigma_G\colon E(G) \to [\ell]$ (the \emph{page assignment}) that assigns each edge %
{%
	$uv \in E(G)$} to a \emph{page} %
{%
	$\sigma_G(uv)$} such that for any two edges %
{%
	$u_1v_1$} and %
{%
	$u_2v_2$} with %
{%
	$\sigma_G(u_1v_1) = \sigma_G(u_2v_2)$} we do not have $u_1v_1 \nestingRelation u_2v_2$. %
If $u_1 \prec u_2 \prec v_1 \prec v_2$, then we say that $u_1v_1$ and $u_2v_2$ form a \emph{twist} in \ql[G]; see also \Cref{fig:example}a.
We write \ql instead of \ql[G] if $G$ is clear from context.
Let $u,v \in V(G)$ be two vertices and \ql a queue layout of $G$.
A vertex $v \in V(G)$ \emph{sees} a vertex $u \in V(G)$ on a page $p \in [\ell]$ if $\langle \prec, \sigma' \rangle$, where %
$\sigma'(uv)=p$ and $\sigma'(e)=\sigma(e)$ otherwise, is a valid queue layout of $G \cup \{uv\}$.
Note that the visibility relation is symmetric.
Furthermore, $v$ is \emph{left} of $u$ if $v \prec u$ and \emph{right} of~$u$ if $u \prec v$.
An \emph{interval} $[u,v]$, with $u \prec v$, is a segment on the spine such that there is no %
$w \in V(G)$ with $u \prec w \prec v$. %
A queue layout \ql[G] \emph{extends} a queue layout \ql[H] of a subgraph $H \subseteq G$ of $G$ if $\sigma_H \subseteq \sigma_G$ and $\prec_H\ \subseteq\ \prec_{G}$, i.e., $\sigma_H(e)=\sigma_G(e)$ for all $e\in E(H)$ and $(u\prec_H v)\Leftrightarrow(u \prec_G v)$ for all $u,v\in V(H)$.

Let $\instance$ be an instance of \QLE.
We use $\Size{\instance}$ as a shorthand for $\Size{V(G)} + \Size{E(G)} + \ell$ and refer to \ql[G], if it exists, as a \emph{solution} of \instance.
For the remainder of this paper, we use $\Vadd \coloneq V(G) \setminus V(H)$ and $\Eadd \coloneqq E(G) \setminus E(H)$ to speak about the set of \emph{new} vertices and edges and denote with $\nadd{}$ and $\madd{}$ their cardinality, respectively.
Vertices and edges of $H$ are called \emph{old}.
We use \emph{missing} as a synonym for \emph{new}. %
Furthermore, $\EaddH{} \coloneqq \{e = uv \in \Eadd{} \mid u, v \in V(H)\}$ is the set of new edges $e$ with only old endpoints $u$ and $v$.

\section{\textsc{QLE} With Only Missing Edges}
\label{sec:only-missing edges}
For our first result, we will use a known equivalence between finding a certain queue layout and a coloring problem in a restricted graph class, combined with existing complexity results on the latter problem.
Dujmović and Wood already note that ``the problem of assigning edges to queues in a fixed vertex ordering is equivalent to colouring a permutation graph'' \cite{DW.LLG.2004} and attribute this result to a 1941 publication by Dushnik and Miller \cite{DM.POS.1941}. %
\ifthenelse{\boolean{long}}{%
    To make this insight more formal and also accessible in modern terms, we want to first independently show this relationship in this section, before later applying it to our problem.
    \par
}{
    In \Cref{app:permutation}, we define the term permutation graph and also make this insight more formal and accessible in modern terms by independently showing this relationship.
}
\begin{statelater}{permutationGraphs}
The \emph{permutation graph} $G_\pi$ corresponding to some permutation $\pi$ of the ordered elements $X$ is the graph that has $X$ as its vertices and where two vertices are adjacent if $\pi$ reverses the order of the corresponding elements. %
To test whether a graph is a permutation graph, we can use the notion of transitive orientations.
A \emph{transitive orientation} $F$ of some undirected graph $G=(V,E)$ assigns directions to all edges in $E$ such that the existence of directed edges $ab\in F$ and $bc\in F$ always also implies an $ac\in F$ edge, oriented in this direction.
Permutation graphs can now be characterized as those graphs where both the graph itself and its complement allow for a transitive orientation (they are thus also called comparability and co-comparability graphs, respectively)~\cite{DM.POS.1941,Gol.AGT.1980}.

We now show the first direction of the claimed statement, i.e., that assigning edges to queues with a fixed spine ordering is equivalent to coloring a certain permutation graph.
Given a graph $G$ and a fixed spine order $\prec_G$, we define the \emph{conflict graph} $C(G,\prec_G)=(E(G),\Cap(\prec_G))$ to be the graph that has $E(G)$ as its vertices and %
edges that correspond to the nesting relation in $\prec_G$.

\begin{lemma}
    \label{lem:queue-layout-permutation-graph}
    For any graph $G$ with fixed spine order $\prec_G$, we have that its conflict graph $C(G,\prec_G)=(E(G),\Cap(\prec_G))$ is a permutation graph.
\end{lemma}
\begin{proof}
    We will explicitly give transitive orientations $F$ and $F'$ for $C(G,\prec_G)$ and its complement $\overline{C(G,\prec_G)}$, respectively.
    Recall that $C(G,\prec_G)$ contains an edge ${a,b}$ if, based on the spine order $\prec_G$, the edge $a$ is nested within edge $b$ or vice versa.
    In $F$, we direct the edge $ab$ from the outer edge to the inner nested edge, or equivalently, according to the distances of the endpoints of $a$ and $b$ in $\prec_G$, from the longer to the shorter edge.
    As the nesting relation is transitive, $F$ is a transitive orientation.

    For $F'$, we will direct ${a,b}\in E(\overline{C(G,\prec_G)})$, i.e., an edge corresponding to two intervals that do not nest, according to their starting points in $\prec_G$, from one that starts earlier to the one that starts later, i.e., from left to right.
    The starts-before relationship is transitive.
    But for $F'$ being transitive, we also have to show that this holds for the restriction to $E(\overline{C(G,\prec_G)})$.
    Assume $ab,bc \in F'$ and denote with $x_s,x_e$ the start and end vertices of the edge $x$ in $\overline{C(G,\prec_G)}$.
    We thus have $a_s \prec_G b_s$ and $b_s \prec_G c_s$ and, as $ab,bc\notin E(C(G,\prec_G))$, i.e., neither $a,b$ nor $b,c$ are in a pairwise nesting relation, we must have
    $a_e \prec_G b_e$, $b_e \prec_G c_e$.
    Together, this yields $a_s \prec_G c_s$ and $a_e \prec_G c_e$ and thus $ac\notin E(C(G,\prec_G))$ and~$ac\notin F'$.
\end{proof}

As the vertices of the conflict graph correspond to the edges of the input graph, and two vertices are adjacent if and only if the corresponding edges are in a nesting relation and thus cannot share a page, the proper $k$-colorings of the vertices of $C(G,\prec_G)$ are exactly the $k$-page queue layouts of $(G,\prec_G)$ by identifying each page with a color (and vice versa).
Having established one direction of the equivalence, we now want to focus on the converse direction.
Here, we will first show that every permutation graph also has a corresponding queue layout instance with a fixed spine order, where edges of the permutation graph are represented through the nesting of edges in the queue layout.

\begin{lemma}
    \label{lem:permutation-graph-queue-layout}
    For every permutation graph $G_\pi$ induced by a permutation $\pi$ of elements $X$,
    the graph $Q(\pi)=(\left\{(v,a) \mid v\in X, a\in\{1,2\}\right\}, \left\{((v,1),(v,2)) \mid v\in X\right\})$ with spine order $\prec_{Q(\pi)}$ of $V(Q(\pi))$ 
    such that $(v_1,a_1) \prec_{Q(\pi)} (v_2,a_2)$ 
        if $a_1 < a_2$, or 
        if $a_1 = a_2 = 1$ and $v_1 < v_2$ in $X$, or
        if $a_1 = a_2 = 2$ and $v_1 \prec v_2$ in $\pi$
    has $\Cap(\prec_{Q(\pi)})=E(G_\pi)$.
\end{lemma}
\begin{proof}
    Note that we have an edge $v_1v_2\in E(G_\pi)$ if and only if the order of $a$ and $b$ is reversed in $\pi$ with regard to their original order in $X$.
    Thus, assuming w.l.o.g.\ $v_1 < v_2$ in $X$, we have that $v_1v_2\in E(G_\pi)$ implies $v_2 \prec v_1$ in $\pi$.
    This means we have $(v_1,1) \prec_{Q(\pi)} (v_2,1) \prec_{Q(\pi)} (v_2,2) \prec_{Q(\pi)} (v_1,2)$ and thus the $v_2$ edge is nested within the $v_1$ edge.
    Conversely, recall that all edges in $Q(\pi)$ are of the form $((v,1),(v,2))$ for some $v\in X$.
    Thus, if $((v_2,1),(v_2,2))$ is nested within $((v_1,1),(v_1,2))$ according to $\prec_{Q(\pi)}$, we have $(v_1,1) \prec_{Q(\pi)} (v_2,1) \prec_{Q(\pi)} (v_2,2) \prec_{Q(\pi)} (v_1,2)$.
    This means that $v_1 < v_2$ in $X$ and $v_2 \prec v_1$ in $\pi$, which implies $v_1v_2\in E(G_\pi)$.
\end{proof}

Note that by the definition of $Q(\pi)$, there is a natural bijection between $X=V(G_\pi)$ and $E(Q(\pi))$.
Moreover, as shown by \Cref{lem:permutation-graph-queue-layout}, two vertices of $G_\pi$ being adjacent and thus needing a different color also means that the corresponding edges in $(Q(\pi), \prec_{Q(\pi)})$ are in a nesting relation, and thus need to be on different pages in a queue layout.
As these implications also hold in the converse direction, the $k$-page queue layouts of $(Q(\pi),\prec_{Q(\pi)})$ are then (by interpreting $k$ pages assigned to $E(Q(\pi))$ as colorings of $X$ with $k$ colors) exactly the $k$-colorings of vertices of $G_\pi$.
Especially, two vertices of $X$ can have the same color in some $k$-coloring if and only if the corresponding edges of $Q(\pi)$ can have the same page in a $k$-page queue layout.
\end{statelater}
\ifthenelse{\boolean{long}}{
    \par
    Altogether, this not only shows equivalence of both problems, but also means that additional constraints on a permutation graph $k$-coloring can equivalently be expressed as additional constraints on a $k$-page-assignment that yields a queue layout.    
}{
    Here, it suffices to note that we not only have equivalence of both problems, but also that additional constraints on a permutation graph $k$-coloring can equivalently be expressed as additional constraints on a $k$-page-assignment with a fixed vertex order that yields a valid queue layout.
}
Especially, if we want a certain $k$-page queue layout, for example one that observes some predefined pages for certain edges, this is equivalent to requiring certain colors for some vertices in the coloring, i.e., an extension of a precoloring.
Using known results from permutation graph coloring, 
that precoloring extension is \NP-complete~\cite{Jan.OCC.1997}
while list-coloring for up to $k$ colors %
is in %
\XP\ parameterized by~$k$~\cite{EST.LCL.2014},
we thus %
obtain the following theorem.

\thmVGisVH*

\newcommand{\stateLemmaOnlyEdgesRemoveSafeText}{\label{lem:only-edges-remove-safe}
Let $\instance = \instanceLong$ be an instance of \QLE with $V(G) = V(H)$ that contains an edge $e \in \Eadd$ with $\Size{P(e)} \geq \madd$.
The instance $\instance' = \left(\ell, G' =  G \setminus \{e\}, H, \ql[H]\right)$ is a positive instance if and only if \instance is a positive instance.}

In light of \Cref{thm:only-edges-hard}, %
we now turn our attention to the parameterized complexity of \QLE when only edges are missing and show that this problem is fixed-parameter tractable in their number.
This matches the complexity of extending stack layouts in the same setting~\cite[Theorem~3.2]{DFGN.PCE.2024}.
In particular, we can use the same proof strategy to obtain the result\ifthenelse{\boolean{long}}{.}{: Remove all new edges $e$ that have at least $\madd$ different pages $p$ such that there is no old edge $e'$ with $\sigma(e') = p $ and $e \nestingRelation e'$, branch to determine the page assignment of the remaining edges, and, finally, check each of the \BigO{{\madd}^{\madd}} different branches in linear time {%
		if it results in} a valid queue layout.}
Nevertheless, for reasons of completeness, we provide the full details %
\ifthenelse{\boolean{long}}{in the following}{of the following theorem in the appendix}.

\ifthenelse{\boolean{long}}{To that end, we define for an edge $e \in \Eadd$, set $P(e) \subseteq [\ell]$ of \emph{admissible pages}, i.e., the set of pages such that $p \in P(e)$ if and only if $\langle\prec_{H},\sigma'_{H}\rangle$ is an $\ell$-page queue layout of $H\cup \{e\}$, where $\sigma'_H(e) = p$ and $\sigma'_H(e') = \sigma_H(e')$ for all edges $e' \in E(H)$.
Intuitively, if only $\madd$ edges are missing from $H$, and a new edge can be placed in many pages, then its concrete page assignment should not influence the existence of a solution.
This intuition is formalized in the following statement.
\begin{lemma}
\stateLemmaOnlyEdgesRemoveSafeText
\end{lemma}
}{}
\begin{prooflater}{plemmaOnlyEdgesRemoveSafe}
	We observe that removing an edge from $G$ and adapting the page assignment $\sigma$ accordingly preserves a solution \ql to \QLE for \instance.
	Thus, we focus for the remainder of the proof on the  ``($\boldsymbol{\Rightarrow}$)-direction''.
	
	\proofsubparagraph*{($\boldsymbol{\Rightarrow}$)}
	
	Let $\instance'$ be a positive instance of \QLE and \ql[G'] a corresponding solution.
	By the definition of~$e$, we can find a page $p \in P(e)$ such that we have $\sigma(e') \neq p$ for every edge $e' \in \Eadd \setminus \{e\}$.
	We now define \ql[G] as $\langle\prec_{G'},\sigma'_{G'}\rangle$, where we extend $\sigma_{G'}$ by the page assignment $\sigma'_{G'}(e) = p$ to $\sigma_{G'}$.
	To see that no two edges in \ql[G] on the same page are in a nesting relation, we first 
	observe that by the definition of $P(\cdot)$ this must be the case for $e$ and any edge $e' \in E(H)$.
	Furthermore, by our selection of $p$, this must also hold between $e$ and any edge $e' \in \Eadd$.
	Hence, \ql[G] is a valid queue layout.
	Finally, since $\ql[G']$ extends $\ql[H]$, we conclude that also \ql[G] is an extension of \ql[H].
	Thus, it witnesses that \instance is a positive instance of \QLE.
\end{prooflater}
\ifthenelse{\boolean{long}}{With \Cref{lem:only-edges-remove-safe} at hand, we can remove all new edges with at least $\madd$ admissible pages. In the following theorem we show that this is sufficient to obtain an algorithm that is fixed-parameter tractable in \madd.}{}
\theoremOnlyEdgesFPT*

\ifthenelse{\boolean{long}}{\begin{prooflater}{ptheoremOnlyEdgesFPT}}{\begin{prooflater}{ptheoremOnlyEdgesFPT}[Proof of \Cref{thm:only-edges-fpt}]}
	First, we compute for every edge $e \in \Eadd$ the set $P(e)$.
	This takes for a single edge linear time, as we can iterate through the old edges $e'$ and check for each of them whether $e \nestingRelation e'$.
	In the spirit of \Cref{lem:only-edges-remove-safe}, we remove $e$ from $G$ if $P(e) \geq \madd$.
	Overall, this takes $\BigO{\madd \cdot\Size{\instance}}$ time and results in a graph $G'$ with $H \subseteq G' \subseteq G$.
	
	Each of the remaining new edges in $G'$ can be assigned to fewer than $\madd$ pages without being in a nesting relation with an old edge.
	Hence, we can brute force their possible page assignments.
	Each of the \BigO{{\madd}^{\madd}} different branches corresponds to a different $\ell$-page queue layouts \ql of $G'$ which, by construction, extends \ql[H].
	We can create \ql in \BigO{\Size{\instance}} time by copying \ql[H] first and augmenting it with $\BigO{\madd}$ new edges.
	For each created $\ell$-page layout, we can check in linear time whether it admits a queue layout, i.e., whether no two edges on the same page are in a nesting relation.
	Note that by our pre-processing step, no new edge can be in a nesting relation with an old edge, and thus it suffices to check whether no pair of new edges $e', e''\in \Eadd \cap E(G')$ on the same page are in a nesting relation.
	If there exists a queue layout for $G'$ that extends \ql[H], then by applying \Cref{lem:only-edges-remove-safe} iteratively, we can complete it to a solution \ql[G] for~\instance.
	Otherwise, we conclude by the same lemma that \instance does not admit the desired $\ell$-page queue layout.
	Combining all, the overall running time is \BigO{{\madd}^{\madd}\cdot\Size{\instance{}}}.
\end{prooflater}

\section{\textsc{QLE} With Missing Vertices and Edges}
\label{sec:qle-kappa}
While \Cref{thm:only-edges-fpt} shows tractability of \QLE when parameterized by the number of missing edges, it only applies to the highly restrictive setting where $H$ already contains all vertices of $G$.
However, in general also vertices are missing from \ql[H] and our task is to find a spine position for them.
We now turn our attention to the complexity of \QLE when both vertices and edges are missing and recall that the problem of deciding whether a graph has a 1-page queue layout is already \NP-complete~\cite{HR.LOG.1992}.
Thus, we %
investigate the complexity of \QLE when parameterized by the number $\kappa \coloneqq \nadd + \madd$ of missing elements.
We show that under this parameter, extending queue and stack layouts admit the same complexity-theoretic behavior: being \XP-tractable but \W[1]-hard~\cite{DFGN.PCE.2024}.%
\subsection{\textsc{QLE} Parameterized by the Number of New Elements is in \textsf{XP}}
We first establish \XP-membership of \QLE when parameterized by $\kappa$ by using an algorithm analogous to the one for extending stack layouts~\cite[Theorem~5.1]{DFGN.PCE.2024}: %
We first branch to determine the spine position of the new vertices.
This gives us $\BigO{\Size{\instance}^{\nadd}}$ different branches, each corresponding to a different spine order $\prec_G$.
The fixed spine order allows us to employ our \FPT-algorithm from \Cref{thm:only-edges-fpt}, resulting in the following theorem. %
\theoremKappaXP*

\begin{prooflater}{ptheoremKappaXP}
	The algorithm that we propose below consists of two steps:
	We first guess the spine position of the new vertices, i.e., determine $\prec_G$.
	This leaves us with an instance of \QLE with $\nadd = 0$, for which we can apply in the second step \Cref{thm:only-edges-fpt}.
	
	Regarding the first step, it is important to note that a solution to \instance could assign multiple new vertices to the same interval in $\prec_{H}$.
	Therefore, we not only branch to determine the spine position of the new vertices but also to fix their order in $\prec_G$.  
	To that end, observe that $\prec_{H}$ induces $\Size{V(H)} + 1$ different intervals, out of which we have to choose $\nadd$ with repetition.
	Considering also the order of the new vertices, we can bound the number of branches by $\nadd!\cdot\binom{\Size{V(H)} + \nadd}{\nadd}$, which is, as the following reformulations underlines, in $\BigO{\Size{\instance}^{\nadd{}}}$.
	\begin{align*}
		\nadd!\cdot\binom{\Size{V(H)} + \nadd}{\nadd} &= \frac{\nadd{}!\cdot (\Size{V(H)} + \nadd{})!}{\nadd{}!\cdot ((\Size{V(H)} + \nadd{}) - \nadd{})!}\\
		&= \frac{(\Size{V(H)} + \nadd{})!}{\Size{V(H)}!}\\
		&= \prod_{i = 1}^{\nadd{}}(\Size{V(H)} + i)\\
		&= \BigO{\Size{\instance}^{\nadd{}}}
	\end{align*}
	
	As discussed above, the spine order $\prec_{G}$ is now fixed and extends $\prec_{H}$.
	Hence, in each branch we only have to check whether $\prec_{G}$ allows for a valid page assignment $\sigma_G$.
	As we have $\nadd= 0$, each branch corresponds to an instance of \QLE with only missing edges.
	Thus, we can use \Cref{thm:only-edges-fpt} to check for a single branch in $\BigO{{\madd}^{\madd}\cdot\Size{\instance}}$ time whether such an assignment $\sigma_G$ exists.
	The overall running time now follows readily.
\end{prooflater}

\Cref{thm:kappa-xp} raises the question whether we can get fixed parameter tractability under $\kappa$.
In what follows, we answer this question negatively and show \W[1]-hardness of \QLE when parameterized by the number of missing elements.

\subsection{\textsc{QLE} Parameterized by the Number of New Elements is Hard}
\label{sec:kappa-w-1}
\newcommand{\stateEdgeGadgetProperty}{%
\begin{restatable}{property}{propertyEdgeGadget}
    \label{prop:edge-gadget}
    Let \ql be a solution to an instance of \QLE with $e = v_{\alpha}^iv_{\beta}^j \in E(G_C)$ and $x_{\alpha},x_{\beta} \in \Vadd$ with $\alpha < \beta$.
	If $u_{\alpha}^1 \prec x_{\alpha} \prec u_{\alpha}^{n_{\alpha} + 1}$, $u_{\beta}^1 \prec x_{\beta} \prec u_{\beta}^{n_{\beta} + 1}$, and $\sigma(x_{\alpha}x_\beta) = p_e$, then $x_{\alpha}$ is placed in  $\intervalPlacing{v_{\alpha}^i}$ and $x_{\beta}$ in $\intervalPlacing{v_{\beta}^j}$.
\end{restatable}}
\newcommand{\stateBotVerticesProperty}{
\begin{restatable}{property}{propertyBotVerticesNotVisible}
    \label{prop:bot-vertices-not-visible}
    For every color $\alpha \in [k]$ and 
    {%
    	for the page $p \in [\ell]$ of every edge gadget}, the vertex $u_{\alpha}^{\bot}$ is not visible on $p$ for any vertex $v$ with $u_1^{\bot_L} \prec v \prec u_{\alpha}^{\bot_L}$ or $u_{\alpha}^{n_{\alpha} + 1} \prec v \prec u_{k + 1}^{\bot_R}$.
\end{restatable}}
\newcommand{\stateFixationGadgetProperty}{
\begin{restatable}{property}{propertyFixationGadget}
    \label{prop:fixation-gadget-property}
    Let \ql be a solution to an instance of \QLE that contains a fixation gadget on $k$ vertices and has \Cref{prop:bot-vertices-not-visible} for every page $p \in [\ell] \setminus \{p_d\}$.
    For every color $\alpha \in [k]$ we have $u_{\alpha}^1 \prec x_{\alpha} \prec u_{\alpha}^{n_{\alpha} + 1}$.
\end{restatable}}
\newcommand{\showEdgeGadgetFigure}{
\begin{figure}[t]
	\centering
	\includegraphics[page=2]{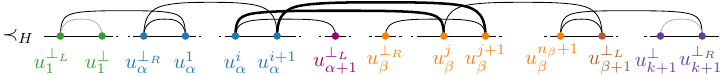}
    \ifthenelse{\boolean{long}}{\caption{Visualization of the edges on the page $p_e$ for the edge $e = v_{\alpha}^iv_{\beta}^j \in E(G_C)$. The {%
    			bold edges form the twist on $p_e$ that we use to synchronize the placement of the vertices and the} gray edges are part of the fixation gadget.}}{\caption{Visualization of the spine order $\prec_H$ and the edge gadget for the edge $e = v_{\alpha}^iv_{\beta}^j \in E(G_C)$ on page $p_e$. The {%
    			bold edges form the twist on $p_e$ that we use to synchronize the placement of the vertices and the} gray edges are part of the fixation gadget.}}
	
	\label{fig:w-1-edge-gadget}
\end{figure}
}
We now show \W[1]-hardness of \QLE when parameterized by $\kappa$.
Our reduction uses the idea and framework of the reduction %
for 
establishing \W[1]-hardness of the problem of extending stack layouts under this parameter~\cite[Theorem~6.4]{DFGN.PCE.2024}.
However, the substantially different notion of visibility makes %
the construction of new gadgets inevitable.
We reduce from \probname{Multi-colored Clique} (\MCC):
Given a graph $G_C$, an integer $k \geq 1$, and a partition of $V(G_C)$ into $k$ independent color sets, i.e., $V(G_C) = V_1\ \dot\cup\ \ldots\ \dot\cup\ V_k$, does there exist a colorful $k$-clique $\mathcal{C}$ in $G_C$?
\MCC is well-known to be \W[1]-hard parameterized by $k$~\cite{CFK+.PA.2015}.

For the remainder of this section, we let $(G_C, k, (V_1, \ldots, V_k))$ be an instance of \MCC with $N = \Size{V(G_C)}$ vertices and $M = \Size{E(G_C)}$ edges.
We use Greek letters to denote the color of a vertex to separate it from its index{%
	, i.e., we let $v_{\alpha}^i$ denote the $i$th vertex of color $\alpha$}.
Based on the instance of \MCC, we construct an instance \instanceLong of \QLE.
\ifthenelse{\boolean{long}}{

In the following, we first give an overview of the reduction  %
before we give a precise description of its construction and establish its correctness. %
}{}
To facilitate description and presentation of our reduction, we allow multi-edges in $H$.
We show in \Cref{app:remove-multi-edges} how we can redistribute the multi-edges over multiple auxiliary vertices while preserving the correctness of the reduction.

\ifthenelse{\boolean{long}}{\begin{figure}
	\centering
	\includegraphics[page=1]{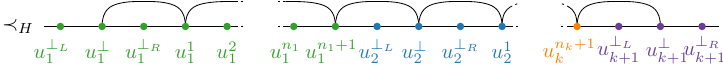}
	\caption{The spine order of $V(H)$ in our reduction and the edges of the fixation gadget on page $p_d$. Vertices of different colors $\alpha \in [k]$ are colored differently. The dummy vertices are colored lilac.}
	\label{fig:w-1-base-layout}
\end{figure}}{}

\subparagraph*{Overview of the Reduction.}
\ifthenelse{\boolean{long}}{}{
{%
At a high level, the constructed instance of \QLE contains $k$ new vertices, one for each color.
For every vertex of $G_C$, we create an interval on the spine $\prec_H$.
We represent membership in $\mathcal{C}$ through the placement of new vertices in the intervals corresponding to vertices in $\mathcal{C}$.
To guarantee that every solution to \MCC corresponds to a solution to \QLE and vice versa, we have to ensure that we select exactly one vertex for each color $\alpha \in [k]$, and that they form a clique. 
For the former requirement, we use a gadget that forces one of the new vertices to be placed in an interval of color $\alpha$.
For the latter property, we let the new vertices form a $k$-clique in $G$ and introduce for each edge $e$ of $G_C$ a further gadget.
This gadget contains a twist that restricts the placement of the endpoints of edges assigned to this page.
In particular, it forces them to lie in the intervals corresponding to the endpoints of $e$, i.e., we only select pairwise adjacent vertices.
With the above intuition in mind, we now describe our reduction in more detail.
}
}

\ifthenelse{\boolean{long}}{}
{\subsubsection{Description of the Reduction.}}
The spine order $\prec_H$ forms the core of our reduction.
To construct it, we introduce for each of the $N$ \emph{original} vertices $v_{\alpha}^i \in V(G_C)$ a \emph{copy} $u_{\alpha}^i \in V(H)$.
Furthermore, we introduce for each color $\alpha$ four dummy vertices $u_{\alpha}^{\bot_L}$, $u_{\alpha}^{\bot}$, $u_{\alpha}^{\bot_R}$, and $u_{\alpha}^{n_{\alpha} + 1}$, with $n_{\alpha} \coloneqq \Size{V_{\alpha}}$, and overall three further dummy vertices $u_{k + 1}^{\bot_L}$, $u_{k + 1}^{\bot}$, and $u_{k + 1}^{\bot_R}$.
\ifthenelse{\boolean{long}}{\Cref{fig:w-1-base-layout}}{\Cref{fig:w-1-edge-gadget}} visualizes the spine order~$\prec_H$, which is the transitive closure of the following orders:
For every $\alpha \in [k + 1]$, we set $u_{\alpha}^{\bot_L} \prec u_{\alpha}^{\bot} \prec u_{\alpha}^{\bot_R}$.
Furthermore, for every $\alpha \in [k]$, we set $u_{\alpha}^{\bot_R} \prec u_{\alpha}^{1}$ and $u_{\alpha}^{n_{\alpha} + 1} \prec u_{\alpha + 1}^{\bot_L}$, and for every $j \in [n_{\alpha}]$ we set $u_{\alpha}^j \prec u_{\alpha}^{j + 1}$.

Observe that for each original vertex $v_{\alpha}^i \in V(G_C)$ we have the interval $[u_{\alpha}^i, u_{\alpha}^{i + 1}]$ on the spine.
We say that the interval $[u_{\alpha}^i, u_{\alpha}^{i + 1}]$ \emph{corresponds} to $v_{\alpha}^i$ and is of color $\alpha$.
We address $[u_{\alpha}^i, u_{\alpha}^{i + 1}]$ with~$\intervalPlacing{v_{\alpha}^i}$.
The graph $G[\Vadd]$ forms a $k$-clique {%
	representing $\mathcal{C}$ in $G$} on the new vertices $\Vadd = \{x_1, \ldots, x_k\}$.
\ifthenelse{\boolean{long}}{To establish an equivalence between the solutions to the individual problems, we must ensure that $\mathcal{C}$ is a clique and colorful.}{}

\ifthenelse{\boolean{long}}{To do the former, we must ensure that two new vertices $x_{\alpha}$ and $x_{\beta}$ can only be placed in $\intervalPlacing{v_{\alpha}^i}$ and $\intervalPlacing{v_{\beta}^j}$, respectively, if $v_{\alpha}^iv_{\beta}^j \in E(G_C)$.
We introduce below the \emph{edge gadget} which enforces this.
For the latter property, each new vertex $x_{\alpha}$ should only be placed in an interval corresponding to a vertex of color $\alpha$.
To ensure this property, we describe the \emph{fixation gadget} afterwards.}{%
To ensure that {%
the selected vertices} form a clique in $G_C$, we introduce for each edge $e = v_{\alpha}^iv_{\beta}^j \in E(G_C)$, $\alpha < \beta$, a so-called \emph{edge gadget}; see \Cref{fig:w-1-edge-gadget}.
This gadget consists of a page $p_e$ on which we place a twist between $\intervalPlacing{v_{\alpha}^i}$ and $\intervalPlacing{v_{\beta}^j}$.
{%
We observe that we introduce one page for each edge of $G_C$.}
{%
The aforementioned} twist ensures that if $x_{\alpha}$ is placed in $\intervalPlacing{v_{\alpha}^i}$, it can see $x_{\beta}$ only if the latter vertex is placed in~$\intervalPlacing{v_{\beta}^j}$.
\showEdgeGadgetFigure
To ensure that the new edges on page $p_e$ can only be used to connect~$x_{\alpha}$ and~$x_{\beta}$ if they are placed in the right intervals, we furthermore introduce edges that block visibility for intervals of color $\gamma$ with $\gamma < \alpha$ (\Cref{fig:w-1-edge-gadget} left), $\alpha \leq \gamma \leq \beta$ (middle), and $\beta < \gamma$ (right).
More formally, we show in \Cref{app:reduction} that the edge gadget has the following properties.
\stateEdgeGadgetProperty
\stateBotVerticesProperty

To ensure that $\mathcal{C}$ is colorful, we must %
force the vertex $x_{\alpha}$ to be placed in an interval of color $\alpha$ for all $\alpha \in [k]$.
We introduce for this purpose the \emph{fixation gadget}.
It consists of $2k$ new edges connecting $x_{\alpha}$ with $u_{\alpha}^\bot$ and $u_{\alpha + 1}^\bot$ for all $\alpha \in [k]$.
These edges are accompanied by the edges $u_{\alpha}^{\bot}u_{\alpha}^1$, $u_{\alpha}^1u_{\alpha}^{n_{\alpha} + 1}$, and $u_{\alpha}^{n_{\alpha} + 1}u_{\alpha + 1}^{\bot}$.
We place these edges on a new page $p_d$ associated to the fixation gadget to ensure that the former vertex can see the latter vertices on $p_d$ only if $u_{\alpha}^1 \prec x_{\alpha} \prec u_{\alpha}^{n_{\alpha} + 1}$.
To reduce visibility on the other pages, we introduce on every page $p \in [\ell] \setminus \{p_d\}$ the edges $u_1^{\bot_L}u_1^{\bot}$ and $u_{k + 1}^{\bot}u_{k + 1}^{\bot_R}$, drawn in gray in \Cref{fig:w-1-edge-gadget}.
As discussed in \Cref{app:reduction}, the gadget has the following property.
\stateFixationGadgetProperty
}

\ifthenelse{\boolean{long}}{
\subparagraph*{Edge Gadget.}}{}
\begin{statelater}{detailsEdgeGadget}
Let $e = v_{\alpha}^iv_{\beta}^j \in E(G_C)$, $\alpha < \beta$ be an edge of $G_C$.
Recall that the $k$ new vertices form a $k$-clique in $G$.
Hence, every pair $x_{\alpha}$ and $x_{\beta}$ of new vertices is adjacent in $G$.
The edge gadget for $e$ consists of a page $p_e$ to which we assign a set of edges to ensure that the new edge $x_{\alpha}x_{\beta}$ can only be assigned to $p_e$ if $x_{\alpha}$ is placed in $\intervalPlacing{v_{\alpha}^i}$ and $x_{\beta}$ in $\intervalPlacing{v_{\beta}^j}$.
{%
	Note that we introduce one page for each edge of $G_C$.}
We create the following edges in $H$, also visualized in \Cref{fig:w-1-edge-gadget}, all of which are assigned to the page $p_e$.

First, we create the edges $u_1^{\bot_L}u_{\alpha}^{1}$ and $u_{\alpha}^{\bot_R}u_{\alpha}^{1}$; see \Cref{fig:w-1-edge-gadget} left.
Second, we create a symmetric set of edges for $\beta$, i.e., $u_{\beta}^{n_{\beta}+1}u_{\beta + 1}^{\bot_L}$ and $u_{\beta}^{n_{\beta}+1}u_{k + 1}^{\bot_R}$; see \Cref{fig:w-1-edge-gadget} right.
\ifthenelse{\boolean{long}}{\showEdgeGadgetFigure}{}
The purpose of these edges becomes clear once we discuss the fixation gadget.
Third, we create the main component of the gadget, visualized in \Cref{fig:w-1-edge-gadget} middle.
It consists of a twist between $\intervalPlacing{v_{\alpha}^i}$ and $\intervalPlacing{v_{\beta}^j}$ with two whiskers on either side to restrict visibility in each interval between $u_{\alpha}^1$ and $u_{\beta}^{n_{\beta} + 1}$.
More concretely, the twist consists of the edges $u_{\alpha}^iu_{\beta}^j$ and $u_{\alpha}^{i + 1}u_{\beta}^{j + 1}$.
The whiskers consist of the edges $u_{\alpha}^{\bot_R}u_{\alpha}^{i + 1}$, $u_{\alpha}^{i}u_{\alpha + 1}^{\bot_L}$, $u_{\beta}^{\bot_R}u_{\beta}^{j + 1}$, and $u_{\beta}^ju_{\beta + 1}^{\bot_L}$.
In concert with the above introduced edges they ensure that $x_{\alpha}$, when placed in $\intervalPlacing{v_{\alpha}^i}$, can see on page $p_e$ only $\intervalPlacing{v_{\beta}^j}$.
Furthermore, the twist, implicitly, removes visibility between $x_{\alpha}$ and $x_{\beta}$ on page $p_e$ if we have $u_{\alpha}^{i + 1} \prec x_{\alpha} \prec x_{\beta} \prec u_{\beta}^j$.

This completes the construction of the gadget and %
we summarize now its properties.
\ifthenelse{\boolean{long}}{\stateEdgeGadgetProperty}{\propertyEdgeGadget*}
\begin{lemma}
    \label{lem:edge-gadget-property}
    Our constructed instance of \QLE has \Cref{prop:edge-gadget}.
\end{lemma}
\begin{proof}
	Let \ql[G] be a solution to \instance for which we have $u_{\alpha}^1 \prec x_{\alpha} \prec u_{\alpha}^{n_{\alpha} + 1}$, $u_{\beta}^1 \prec x_{\beta} \prec u_{\beta}^{n_{\beta} + 1}$, and $\sigma(x_{\alpha}x_\beta) = p_e$ for $e = v_{\alpha}^iv_{\beta}^j \in E(G_C)$.
	In the following, we first show that this implies $u_{\alpha}^i \prec x_{\alpha}$ and $x_{\alpha} \prec u_{\alpha}^{i + 1}$.
	Thus, we show that the only feasible position for $x_{\alpha}$ that remains is $\intervalPlacing{v_{\alpha}^i}$.
	Once this has been established, the twist on page $p_e$ implies $x_{\beta}$ in $\intervalPlacing{v_{\beta}^j}$.

	We first observe that $u_{\alpha}^1 \prec x_{\alpha} \prec u_{\alpha}^i$ gives us no placement for $x_{\beta}$ that is compatible with $\sigma(x_{\alpha}x_\beta) = p_e$:
	As $\alpha < \beta$, we obtain $x_{\alpha} \prec u_{\alpha}^i \prec u_{\alpha + 1}^{\bot_L} \prec x_{\beta}$.
	Since $\sigma(u_{\alpha}^iu_{\alpha + 1}^{\bot_L}) = p_e$, we obtain that $x_{\alpha}x_{\beta}$ nests the edge $u_{\alpha}^iu_{\alpha + 1}^{\bot_L}$ on page $p_e$.
	Hence, we must have $u_{\alpha}^{i} \prec x_{\alpha}$.

	Next, we show that $u_{\alpha}^{i + 1} \prec x_{\alpha}$ is impossible due to the absence of a placement of $x_{\beta}$ that avoids a nesting relation on $p_e$.
	First, we observe that $x_{\beta} \prec u_{\beta}^{j + 1}$ is impossible due to $\sigma(u_{\alpha}^{i + 1}u_{\beta}^{j + 1}) = p_e$, as we have $u_{\alpha}^{i + 1} \prec x_{\alpha} \prec x_{\beta} \prec u_{\beta}^{j + 1}$.
	Hence, from $u_{\alpha}^{i + 1} \prec x_{\alpha}$, the only placement that is still possible is $u_{\beta}^{j + 1} \prec x_{\beta}$.
	However, having $u_{\beta}^{j + 1} \prec x_{\beta}$ gives us $x_{\alpha} \prec u_{\beta}^{\bot_R} \prec u_{\beta}^{j + 1} \prec x_{\beta}$.
	Hence, $x_{\alpha}x_{\beta}$ nests $u_{\beta}^{\bot_R}u_{\beta}^{j + 1}$ and since both edges are on $p_e$, we conclude that also $u_{\beta}^{j + 1} \prec x_{\beta}$ is impossible.
	This brings us to the conclusion that $u_{\alpha}^{i + 1} \prec x_{\alpha}$ is also impossible.
	We thus obtain $u_{\alpha}^{i} \prec x_{\alpha} \prec u_{\alpha}^{i + 1}$, i.e., $x_{\alpha}$ is placed in $\intervalPlacing{v_{\alpha}^i}$.
	
	To complete the proof of this lemma, observe that $x_{\alpha}$ in $\intervalPlacing{v_{\alpha}^i}$, $x_{\alpha} \prec x_{\beta}$, and $\sigma(x_{\alpha}x_{\beta}) = p_e$ implies $x_{\beta}$ in $\intervalPlacing{v_{\beta}^j}$ thanks to the twist on the page $p_e$:
	$x_{\beta} \prec u_{\beta}^j$ implies $u_{\alpha}^i \prec x_{\alpha} \prec x_{\beta} \prec u_{\beta}^j$, which leads to $x_{\alpha}x_{\beta}$ being nested by the edge $u_{\alpha}^iu_{\beta}^j$, and $u_{\beta}^{j + 1} \prec x_{\beta}$ implies $x_{\alpha} \prec u_{\alpha}^{i + 1} \prec u_{\beta}^{j + 1} \prec x_{\beta}$, which causes the edge $u_{\alpha}^{i + 1}u_{\beta}^{j + 1}$ to be nested by $x_{\alpha}x_{\beta}$.
	Hence, we have $u_{\beta}^j \prec x_{\beta} \prec u_{\beta}^{j + 1}$.
	Combining all, we conclude that $x_{\alpha}$ is placed in $x_{\alpha}$ in $\intervalPlacing{v_{\alpha}^i}$ and $x_{\beta}$ is placed in $\intervalPlacing{v_{\beta}^j}$.%
\end{proof}

\ifthenelse{\boolean{long}}{\stateBotVerticesProperty}{\propertyBotVerticesNotVisible*}
\begin{lemma}
	\label{lem:bot-vertices-not-visible}
	Our constructed instance of \QLE has \Cref{prop:bot-vertices-not-visible}.
\end{lemma}
\begin{proof}
	Consider a vertex $v$ such that $u_1^{\bot_L} \prec v \prec u_{\alpha}^{\bot_L}$ and a page $p_e \in [\ell] \setminus \{p_d\}$ for an edge $e = v_{\beta}^iv_{\gamma}^j \in E(G_C)$.
	Without loss of generality we assume $\beta < \gamma$.
	We can identify the following cases depending on the relative order of $v$ with respect to the copies of the endpoints of the edge $e$.
	
	\proofsubparagraph*{Case 1: $v \prec u_{\beta}^i \prec u_{\gamma}^j$.}
	There are two sub-cases we can consider: $\alpha \leq \beta$ and $\beta < \alpha$.
	For the first sub-case, i.e., $\alpha \leq \beta$, it is sufficient to observe that we have $\sigma(u_1^{\bot_L}u_{\beta}^1) = p_e$ and (thanks to $v \prec u_{\alpha}^{\bot_L}$ for $\alpha = \beta$)  $u_1^{\bot_L} \prec v \prec u_{\alpha}^{\bot} \prec u_{\beta}^1 \preceq u_{\beta}^i$.
	Hence, already the existence of the edge $u_1^{\bot_L}u_{\beta}^1$ on this page blocks visibility to $u_{\alpha}^{\bot}$ for $v$.
	For the second sub-case, i.e., $\beta < \alpha$, we immediately get the desired property by combining $v \prec u_{\beta}^i \prec u_{\beta + 1}^{\bot_L} \prec u_{\alpha}^{\bot}$ and $\sigma(u_{\beta}^iu_{\beta + 1}^{\bot_L}) = p_e$.\looseness=-1
		
	\proofsubparagraph*{Case 2: $u_{\beta}^i \prec v \prec u_{\gamma}^j$.}
	Recall that we have $u_1^{\bot_L} \prec v \prec u_{\alpha}^{\bot_L}$.
	Hence, $u_{\beta}^i \prec v \prec u_{\gamma}^j$ implies $\beta < \alpha$, which leaves us this time with the following two sub-cases:
	$\alpha \leq \gamma$ and $\gamma < \alpha$.
	In the first sub-case, i.e., $\alpha \leq \gamma$, we make use of the edge $u_{\beta}^iu_{\gamma}^j$.
	Combining the prerequisites of the sub-case, we get $u_{\beta}^i \prec v \prec u_{\alpha}^{\bot_L} \prec u_{\alpha}^{\bot} \prec u_{\gamma}^j$.
	Hence, the edge $u_{\beta}^iu_{\gamma}^j$ on page $p_e$ would nest with $x_{\alpha}u_{\alpha}^{\bot}$, thus preventing $v$ from seeing $u_{\alpha}^{\bot}$ on $p_e$.
	The second sub-case, i.e., $\gamma < \alpha$, yields $v \prec u_{\gamma}^j \prec u_{\gamma}^{n_{\gamma}+1} \prec u_{\alpha}^{\bot}$.
	Hence, the edge $u_{\gamma}^{j}u_{\gamma}^{n_{\gamma} + 1}$ with $\sigma(u_{\gamma}^{j}u_{\gamma}^{n_{\gamma} + 1}) = p_e$ is sufficient to prevent visibility for $v$ to $u_{\alpha}^\bot$.\looseness=-1
	
	\proofsubparagraph*{Case 3: $u_{\beta}^i \prec u_{\gamma}^j \prec v$.}
	For the third and last case we obtain the desired property almost immediately.
	Observe that $u_{\beta}^i \prec u_{\gamma}^j \prec v$ and $v \prec u_{\alpha}^{\bot_L}$ imply $u_{\gamma}^j \prec v \prec u_{\alpha}^{\bot} \prec u_{k + 1}^{\bot_R}$.
	If we have $v \prec u_{\gamma}^{n_{\gamma} + 1}$, we can use the edge $u_{\gamma}^{n_{\gamma}+1}u_{\gamma + 1}^{\bot_L}$ on page $p_e$ together with the fact that $u_{\gamma + 1}^{\bot_L} \prec u_{\gamma + 1}^{\bot} \preceq u_{\alpha}^{\bot}$ which gives us $u_{\gamma}^j \prec v \prec u_{\gamma}^{n_{\gamma} + 1} \prec u_{\gamma + 1}^{\bot_L} \prec u_{\alpha}^{\bot}$.
	Thus, $v$ cannot see $u_{\alpha}^{\bot}$ on $p_e$.
	Otherwise, i.e., if we have $u_{\gamma}^{n_{\gamma} + 1} \prec v$, the existence of the edge $u_{\gamma}^{n_{\gamma}+1}u_{k + 1}^{\bot_R}$ on page $p_e$ together with $u_{\gamma}^{n_{\gamma} + 1} \prec v \prec u_{\alpha}^{\bot} \prec u_{k + 1}^{\bot_R}$ is sufficient to conclude that $u_{\alpha}^{\bot}$ is not visible for $v$ on page $p_e$.\looseness=-1
	
	Combining all three cases, we conclude that our construction has \Cref{prop:bot-vertices-not-visible} if $u_1^{\bot_L} \prec v \prec u_{\alpha}^{\bot_L}$.
	The setting with $u_{\alpha}^{n_{\alpha} + 1} \prec v \prec u_{k + 1}^{\bot_R}$ can be shown by an analogous, symmetric case analysis.
\end{proof}

\end{statelater}
\ifthenelse{\boolean{long}}{\subparagraph*{Fixation Gadget.}}{}
\begin{statelater}{detailsFixationGadget}
Observe that \Cref{prop:edge-gadget} requires the new vertices to be placed between $u_{\alpha}^1$ and $u_{\alpha}^{n_{\alpha}  + 1}$ and $u_{\beta}^1$ and $u_{\beta}^{n_{\beta}  + 1}$, respectively.
Furthermore, recall that a colorful clique must contain one vertex from each color.
Hence, the correctness of our reduction hinges on a way to restrict the placement of new vertices to be between specific old vertices.
The fixation gadget, which we describe in the following and visualize in \ifthenelse{\boolean{long}}{\Cref{fig:w-1-base-layout}}{\Cref{fig:w-1-example}}, gives us a way to achieve this.

The gadget consists of $2k$ new edges $\{x_{\alpha}u_{\alpha}^\bot, x_{\alpha}u_{\alpha + 1}^\bot \mid \alpha \in [k]\} \subset \Eadd$.
Thanks to a set of old edges, which we introduce below, $x_{\alpha}$ can see $u_{\alpha}^\bot$ and $u_{\alpha + 1}^\bot$ on a special page $p_d$ only if placed between $u_{\alpha}^{1}$ and $u_{\alpha}^{n_{\alpha} + 1}$.
Thus, these edges restrict the placement of $x_{\alpha}$ if they are assigned to $p_d$.

First, we introduce for every color $\alpha \in [k]$ the following old edges on page $p_d$.
The edge $u_{\alpha}^{\bot}u_{\alpha}^1$ connects $u_{\alpha}^{\bot}$ with $u_{\alpha}^1$ and prevents visibility for $x_{\alpha}$ to $u_{\alpha + 1}^{\bot}$ if $x_{\alpha} \prec u_{\alpha}^{\bot}$.
The next edge, $u_{\alpha}^1u_{\alpha}^{n_{\alpha} + 1}$, limits the visibility to $u_{\alpha + 1}^{\bot}$ further to $x_{\alpha} \prec u_{\alpha}^{1}$.
Additionally, if $u_{\alpha}^{n_{\alpha} + 1} \prec x_{\alpha}$, the new vertex $x_{\alpha}$ can no longer see its neighbor $u_{\alpha}^{\bot}$ on page $p_d$.
Next, we introduce the edge $u_{\alpha}^{n_{\alpha} + 1}u_{\alpha + 1}^{\bot}$ whose effect is symmetric to the one of the first edge.
Finally, we introduce on every page $p \in [\ell]$, $p \neq p_d$ the edges $u_1^{\bot_L}u_1^{\bot}$ and $u_{k + 1}^{\bot}u_{k + 1}^{\bot_R}$, visualized in gray in \Cref{fig:w-1-edge-gadget}.

As indicated above, these edges should allow visibility to $u_{\alpha}^\bot$ and $u_{\alpha + 1}^\bot$ on page $p_d$ only if $u_{\alpha}^{1} \prec x_{\alpha} \prec u_{\alpha}^{n_{\alpha} + 1}$.
We now summarize this property.
\ifthenelse{\boolean{long}}{\stateFixationGadgetProperty}{\propertyFixationGadget*}
\begin{lemma}
    Our constructed instance of \QLE has \Cref{prop:fixation-gadget-property}
\end{lemma}
\begin{proof}
	Assume, for the sake of a contradiction, $x_{\alpha} \prec u_{\alpha}^1$ and consider the page on which the edge $x_{\alpha}u_{\alpha + 1}^{\bot}$ is placed, i.e., the page $p \in [\ell]$ with $\sigma(x_{\alpha}u_{\alpha + 1}^{\bot}) = p$.
	First, we observe that $x_{\alpha} \prec u_1^{\bot_L}$ is impossible, since we have $\sigma(u_1^{\bot}u_1^1) = p_d$ and $\sigma(u_1^Lu_1^{\bot}) = p$ for every $p \in [\ell] \setminus \{p_d\}$.
	Since we have $x_{\alpha} \prec u_1^{\bot_L} \prec u_1^{\bot} \prec u_1^1 \prec u_{\alpha + 1}^{\bot}$, the edge $x_{\alpha}u_{\alpha + 1}^{\bot}$ is on every page $p \in [\ell]$ in a nesting relation with some edge on $p$. 
    Hence, we have $u_1^{\bot_L} \prec x_{\alpha} \prec u_{\alpha}^1$ and since we have $\sigma(u_{\alpha}^1u_{\alpha}^{n_{\alpha}+1}) = p_d$ and $u_1^{\bot_L} \prec x_{\alpha} \prec u_{\alpha}^1 \prec u_{\alpha}^{n_{\alpha}+1} \prec u_{\alpha + 1}^\bot$, we must have $p \neq p_d$.
    Until now we have $u_1^{\bot_L} \prec x_{\alpha} \prec u_{\alpha + 1}^{\bot_L}$.
    \Cref{prop:bot-vertices-not-visible} now tells us that $u_{\alpha + 1}^{\bot}$ is not visible on $p$ from any spine position left of $u_{\alpha + 1}^{\bot_L}$ and in particular not from the position of $x_{\alpha}$.
    Since we cannot find a page $p$ on which the edge $x_{\alpha}u_{\alpha + 1}^{\bot}$ can be placed, we conclude that $x_{\alpha} \prec u_{\alpha}^1$ cannot be the case.
	By similar arguments we can also derive that $u_{\alpha}^{n_{\alpha} + 1} \prec x_{\alpha}$ is impossible.
	Thus, we conclude $u_{\alpha}^1 \prec x_{\alpha} \prec u_{\alpha}^{n_{\alpha} + 1}$.
\end{proof}

\end{statelater}
\ifthenelse{\boolean{long}}{To complete our reduction, we now combine both gadgets.}{}

\ifthenelse{\boolean{long}}{
\subparagraph*{The Complete Reduction.}
Recall that we reduce from \MCC.
The strategy of our reduction is to insert $k$ new vertices into our layout that form a $k$-clique in $G[\Vadd]$, where $k$ is the parameter for an instance $(G_C, k, (V_1, \ldots, V_k))$ of the said problem.
To obtain the correspondence between solutions to the two problems, we created for each vertex $v_{\alpha}^i \in G_C$ a copy $u_{\alpha}^i$ and defined the spine order $\prec_H$ based on the color and index of the vertices.
This gives us for each $v_{\alpha}^i \in G_C$ an interval in $\prec_H$ that corresponds to it.
We introduced the fixation gadget that forces each new vertex to be placed in some interval that corresponds to a vertex of its color.
Finally, we have one edge gadget for every edge in $G_C$.
Thus, there is one page for each edge in $G_C$.
The twists on these pages ensure that two new vertices can only be placed in intervals corresponding to adjacent vertices in $G_C$, effectively selecting a colorful clique in $G_C$.

\Cref{fig:w-1-example} shows a small illustrative example of our reduction.
}{
\subsubsection{Correctness of the Reduction.}
The above (informal) description of our construction should only convey its main idea.
We provide in \Cref{app:reduction} a precise description of the gadgets and show a small illustrative example in \Cref{fig:w-1-example}.
}
\begin{figure}[t]
	\centering
	\includegraphics[page=2]{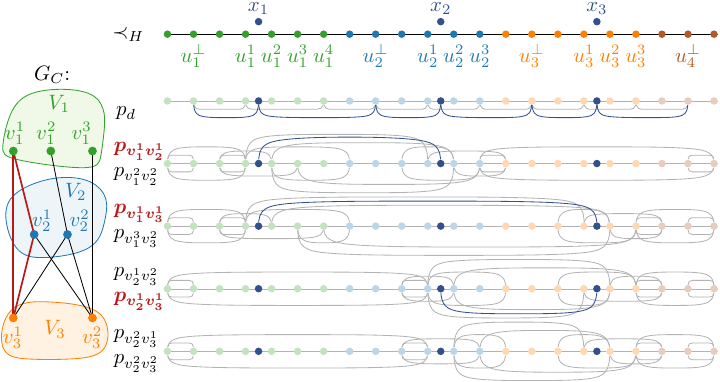}
	\caption{An example of the reduction. We indicate the solution to \MCC in red (left) and the corresponding solution to the instance \instanceLong of \QLE in saturated colors (right). To increase readability, we omit the labels for the vertices $u_{\alpha}^{\bot_L}$ and $u_{\alpha}^{\bot_R}$.}
	\label{fig:w-1-example}
\end{figure}%

Observe that even for $u_{\alpha}^1 \prec x_{\alpha} \prec u_{\alpha}^{n_{\alpha} + 1}$, the edges of the edge gadget prevent visibility to $u_{\alpha}^{\bot}$ and $u_{\alpha + 1}^{\bot}$.
Thus, the edges $x_{\alpha}u_{\alpha}^\bot$ and $x_{\alpha}u_{\alpha + 1}^\bot$ can only be placed on the dedicated page for the fixation gadget.
Furthermore, for every new edge~$x_{\alpha}x_{\beta}$ in $G[\Vadd]$ we have $\alpha \neq \beta$ and since $\sigma(u_{\alpha}^{\bot}u_{\alpha}^1) = \sigma(u_{\alpha}^{n_{\alpha} + 1}u_{\alpha + 1}^{\bot}) = p_d$, we conclude that these are the only edges on $p_d$.
More formally, we obtain:

\begin{observation}
	\label{obs:fixation-gadget-edges}
	Let \ql be a solution to our created instance of \QLE.
	For every $\alpha \in [k]$ we have $\sigma(x_{\alpha}u_{\alpha}^\bot) = \sigma(x_{\alpha}u_{\alpha + 1}^\bot) = p_d$.
    Furthermore, for every edge $e \in \Eadd$ with $\sigma(e) = p_d$ we have $e \in \{x_{\alpha}u_{\alpha}^\bot, x_{\alpha}u_{\alpha + 1}^\bot \mid \alpha \in [k]\}$.
\end{observation}
Furthermore, observe that the page $p_e$ for an edge $e = v_{\alpha}^iv_{\beta}^j \in E(G_C)$ can only be used for new edges that have their endpoints placed in intervals of color~$\alpha$ and $\beta$.
If at least one endpoint is placed in an interval of a different color $\gamma$, we create, depending on the relation of $\alpha$, $\beta$, and $\gamma$, always a nesting relation with at least one edge on $p_e$.
The following observation summarizes this property.
\begin{observation}
	\label{obs:edge-gadget-colors}
	Let \ql be a solution to our constructed instance of \QLE.
        For every new edge $x_{\gamma}x_{\delta} \in E(G[\Vadd])$ assigned to the page $p_e \in [\ell] \setminus \{p_d\}$ with $e = v_{\alpha}^iv_{\beta}^j \in E(G_C)$, $\alpha < \beta$, and $\gamma < \delta$, we have $\alpha = \gamma$ and $\beta = \delta$.
\end{observation}
We now have all the ingredients to show the main theorem of this section.

\theoremKappaWHard*
\begin{proofsketch}
    Let $(G_C, k, (V_1, \ldots, V_k))$ be an instance of \MCC and $\instance$ the instance of \QLE obtained by our reduction.
    A careful analysis of \instance reveals that we have $\Size{\instance} = \BigO{\Size{G_C} + k^2}$ and $\kappa = \BigO{k^2}$.
    In particular, $\nadd = k$ and $\madd = \binom{k}{2} + 2k$.
    We now proceed with discussing correctness of our reduction.

    For the forward direction, let $\mathcal{C}$ be a colorful $k$-clique.
    To construct a solution of \instance, we consider every vertex $v_{\alpha}^i \in \mathcal{C}$ and place $x_{\alpha}$ in $\intervalPlacing{v_{\alpha}^i}$.
    As $\mathcal{C}$ is colorful, this placement is well-defined and yields the spine order $\prec_G$.
    For the page assignment $\sigma_G$, we set $\sigma_G(x_{\alpha}u_{\alpha}^{\bot}) = \sigma_G(x_{\alpha}u_{\alpha + 1}^{\bot}) = p_d$ for every $\alpha \in [k]$.
    As this resembles the necessary page assignment from \Cref{obs:fixation-gadget-edges}, we do not introduce a nesting relation on page $p_d$.
    Furthermore, for each new edge $x_{\alpha}x_{\beta} \in E(G[\Vadd])$, we consider the intervals $\intervalPlacing{v_{\alpha}^i}$ and $\intervalPlacing{v_{\beta}^j}$ into which $x_{\alpha}$ and $x_{\beta}$ are placed.
    As $\mathcal{C}$ is a clique, there exists the edge $e = v_{\alpha}^iv_{\beta}^j \in E(G_C)$ and consequently the page $p_e$ in \instance that can accommodate $x_{\alpha}x_{\beta}$.
    This shows that \ql[G] is a solution to \instance.
    
    For the backward direction, let \ql[G] be a solution to \instance.
    Thanks to \Cref{prop:fixation-gadget-property,prop:bot-vertices-not-visible}, we have $u_{\alpha}^1 \prec_G x_{\alpha} \prec_G u_{\alpha}^{n_{\alpha + 1}}$ for each new vertex $x_{\alpha}$.
    As each interval between $u_{\alpha}^1$ and $u_{\alpha}^{n_{\alpha + 1}}$ corresponds to a vertex of $G_C$, $x_{\alpha}$ is placed in~$\intervalPlacing{v_{\alpha}^i}$ for some $v_{\alpha}^i \in V(G_C)$.
    We add $v_{\alpha}^i$ to our solution $\mathcal{C}$ of the instance of \MCC and observe that it is colorful. 
    To see that $\mathcal{C}$ is a clique in $G_C$, we consider the page $p$ for a new edge $x_{\alpha}x_{\beta} \in E(G[\Vadd])$.
    \Cref{obs:fixation-gadget-edges} ensures $p = p_e$ for some $e \in E(G_C)$ and
    \Cref{prop:edge-gadget} and \Cref{obs:edge-gadget-colors} guarantee that $x_{\alpha}$ and $x_{\beta}$ are placed in the intervals that correspond to the endpoints of $e$. %
    Extending this to all new edges $E(G[\Vadd])$, we conclude that $\mathcal{C}$ is also a clique in $G_C$.
\end{proofsketch}
\begin{prooflater}{ptheoremKappaWHard}
	Let $(G_C, k, (V_1, \ldots, V_k))$ be an instance of \MCC and $\instance$ the instance of \QLE obtained by our reduction with $\instance = \instanceLong$.
	Let $G_C$ consist of $N = \Size{V(G_C)}$ vertices and $M = \Size{E(G_C)}$ edges.
	Recall that \instance contains $k$ missing vertices that form a $k$-clique and the fixation gadget furthermore contributes $2k$ missing edges.
	Therefore, we have $\kappa = \nadd + \madd = k + \binom{k}{2} + 2k = 3k + \binom{k}{2}$.
	Regarding the size of \instance, we observe that $H$ consists of $N + 4k + 3$ vertices.
	Furthermore, we create for each edge of $G_C$ one page, on which we place up to ten edges.
	The fixation gadget introduces one additional page with three edges per color and two edges per remaining page in the construction, which are $M$.
	Summing this up, we obtain $12M + 3k$ edges in $H$ distributed among $\ell = M + 1$ pages.
	Together with $\kappa = \BigO{k^2}$, we conclude that the size of \instance is in $\BigO{N + M + k^2}$.
	This is polynomial in the size of $G_C$ and the parameter $\kappa$ is bounded by a function in $k$.
	\instance can be created in time polynomial in $N + M + k^2$ and it thus remains to show correctness of the reduction.
    \looseness=-1
	
	\proofsubparagraph*{($\Rightarrow$)}
	Let $(G_C, k, (V_1, \ldots, V_k))$ be a positive instance of \MCC and $\mathcal{C}$ a solution that witnesses it.
	We now construct based on $\mathcal{C}$ a queue layout \ql[G] of $G$ that extends \ql[H].
	To that end, we consider every vertex $v_{\alpha}^i \in \mathcal{C}$ and place the new vertex $x_{\alpha}$ of color $\alpha$ next to the copy $u_{\alpha}^i$, i.e., in the corresponding interval $\intervalPlacing{v_{\alpha}^i}$.
	The final spine order $\prec_G$ is obtained by taking the transitive closure of the order defined above.
	Since $\mathcal{C}$ is colorful, it contains exactly one vertex of each color, i.e., for each color $\alpha$ we have $u_{\alpha}^1 \prec_G x_{\alpha} \prec_G u_{\alpha}^{n_{\alpha} + 1}$.
	Since we only insert the $k$ missing vertices, the spine order $\prec_G$ clearly extends $\prec_H$.
	We can obtain the page assignment $\sigma_G$ by copying $\sigma_H$, to ensure that the former extends the latter, and setting $\sigma_G(x_{\alpha}u_{\alpha}^{\bot}) = \sigma_G(x_{\alpha}u_{\alpha + 1}^{\bot}) = p_d$ for every $\alpha \in [k]$.
	Furthermore, for every new edge $x_{\alpha}x_{\beta} \in \Eadd$ with $1 \leq \alpha < \beta \leq k$ we consider the interval the two endpoints are placed in, let it be $\intervalPlacing{v_{\alpha}^i}$ and $\intervalPlacing{v_{\beta}^j}$.
	As $\mathcal{C}$ is a clique, and thanks to the way we place these new vertices, we know that $v_{\alpha}^iv_{\beta}^j = e \in E(G_C)$.
	We set $\sigma_G(x_{\alpha}x_{\beta}) = p_e$ and know by above arguments that the page $p_e$ must exist in \instance.
	This completes the construction of \ql[G].
	By construction, it extends \ql[H].
	Thus, it only remains to show that no two edges on the same page are in a nesting relation.
	Clearly, no two new edges of $\Eadd$ can be in a nesting relation since they are either put on different pages, if they are of the form $x_{\alpha}x_{\beta}$, $\alpha, \beta \in [k]$, or their endpoints appear consecutively on the spine, i.e., we have $u_1^{\bot} \prec_G x_1 \prec_G u_2^{\bot} \prec_G x_2 \prec_G \ldots \prec_G x_k \prec_{G} u_{k + 1}^{\bot}$.
	Hence, we only need to ensure that no pair of new and old edge can be in a nesting relation.
	To see that this is the case for a new edge from the fixation gadget, it is sufficient to observe that we resemble in \ql[G] the necessary spine order and page assignment from \Cref{prop:fixation-gadget-property} and \Cref{obs:fixation-gadget-edges}.
	Finally, to show that this is also the case for a new edge between new vertices, i.e., for edges of the form $x_{\alpha}x_{\beta}$, for $1 \leq \alpha < \beta \leq k$, we consider the page $\sigma_G(x_{\alpha}x_{\beta}) = p_e$ for the edge $v_{\alpha}^iv_{\beta}^j$.
	We can see that the whiskers on page $p_e$ are defined on old vertices $u, u' \in V(H)$ such that we have $u \prec_G x_{\alpha} \prec_G u' \prec_G x_{\beta}$, $u \preceq_G u_{\alpha}^i \prec_G x_{\alpha} \prec_G u' \prec_G x_{\beta}$, $x_{\alpha} \prec_G u \prec_G x_{\beta} \prec_G u_{\beta}^{j+1}$, or $x_{\alpha} \prec_G u \preceq_G u_{\beta}^i \prec_G x_{\beta} \prec_G u'$.
	Thus, a nesting relation is not possible here.
	For the twist, we recall that it consists of the two edges $u_{\alpha}^iu_{\beta}^j$ and $u_{\alpha}^{i + 1}u_{\beta}^{j + 1}$ and we have $u_{\alpha}^i \prec_G x_{\alpha} \prec_G u_{\alpha}^{i + 1} \prec_G u_{\beta}^j \prec_G x_{\beta} \prec_G u_{\beta}^{j + 1}$, i.e., also here there is no nesting relation possible.
	Note that all edges are of one of the above forms and since in neither case a nesting relation between edges on the same page is possible, we derive that our layout is indeed a valid queue layout.
	
	Hence, combing above arguments, we conclude that \ql[G] is a queue layout that extends \ql[H], i.e., it witnesses that \instance is a positive instance of \QLE.
	
	\proofsubparagraph*{($\Leftarrow$)}
	Let \instance be a positive instance of \QLE and \ql[G] an extension of \ql[H] that witnesses this.
	We now use $\prec_G$ to determine which $k$ vertices of $G_C$ form the colorful clique.
	To that end, observe that thanks to \Cref{prop:fixation-gadget-property,prop:bot-vertices-not-visible} we have for each color $\alpha \in [k]$ that the new vertex $x_{\alpha}$ is placed between $u_{\alpha}^1$ and $u_{\alpha}^{n_{\alpha + 1}}$, i.e., we have $u_{\alpha}^1 \prec_G x_{\alpha} \prec_G u_{\alpha}^{n_{\alpha + 1}}$.
	Hence, $x_{\alpha}$ is placed in one interval that corresponds to a vertex $u_{\alpha}^i \in V_{\alpha} \subseteq V(G_C)$.
	We let $\mathcal{C}$ consist of precisely the vertices that correspond to these intervals, i.e., if $x_{\alpha}$ is placed in $\intervalPlacing{v_{\alpha}^i}$, we add $v_{\alpha}^i$ to $\mathcal{C}$.
	Observe that \Cref{prop:fixation-gadget-property} guarantees that $\mathcal{C}$ contains exactly one vertex $v_{\alpha}^i$ for each color $\alpha \in [k]$, i.e., is colorful.
	What remains to show is that $\mathcal{C}$ forms a clique in $G_C$.
	To that end, let $v_{\alpha}^i, v_{\beta}^j \in \mathcal{C}$ be two arbitrary vertices.
	In the following, we show $v_{\alpha}^iv_{\beta}^j \in E(G_C)$ and assume, without loss of generality, $\alpha < \beta$.
	From $v_{\alpha}^i, v_{\beta}^j \in \mathcal{C}$ we conclude that $x_{\alpha}$ is placed in $\intervalPlacing{v_{\alpha}^i}$ and $x_{\beta}$ is placed in $\intervalPlacing{v_{\beta}^j}$.
	Recall $x_{\alpha}x_{\beta} \in \Eadd$.
	Hence, we have $\sigma_G(x_{\alpha}x_{\beta}) = p$ for some page $p \in [\ell]$ and we can differentiate between the following cases.
	First, observe that $p = p_d$, i.e., the page from the fixation gadget, is not possible because of \Cref{obs:fixation-gadget-edges}.
	Consequently, $p \in \{p_e \mid e \in E(G_C)\}$ holds.
    From \Cref{obs:edge-gadget-colors}, we get that the endpoints of $e$ are of color $\alpha$ and $\beta$.
    Furthermore, observe that we have $u_{\alpha}^{1} \prec_G x_{\alpha} \prec_{G} u_{\alpha}^{n_{\alpha} + 1}$ and $u_{\beta}^{1} \prec_G x_{\beta} \prec_G u_{\beta}^{n_{\beta}} + 1$.
	Thus, all prerequisites of \Cref{prop:edge-gadget} are fulfilled, which tells us that from $\sigma_G(x_{\alpha}x_{\beta}) = p_e$, $e \in E(G_C)$, we have that $x_{\alpha}$ and $x_{\beta}$ are placed in the intervals corresponding to the (copy) of the endpoints of $e$.
	Since we only create pages for edges of $G_C$, we derive from $\sigma_G(x_{\alpha}x_{\beta}) = p_e$, $x_{\alpha}$ being placed in $\intervalPlacing{v_{\alpha}^i}$, and $x_{\beta}$ in $\intervalPlacing{v_{\beta}^j}$, that $v_{\alpha}^i$ and $v_{\beta}^j$ are adjacent in $G_C$, i.e., $v_{\alpha}^iv_{\beta}^j \in E(G_C)$.
	Since $v_{\alpha}^i$ and $v_{\beta}^j$ are arbitrary vertices from $\mathcal{C}$, we conclude that $\mathcal{C}$ forms a clique in $G_C$.
	We conclude that $\mathcal{C}$ forms a colorful clique in $G_C$, i.e., $(G_C, k, (V_1, \ldots, V_k))$ is a positive instance of \MCC.
\end{prooflater}

\section{Adding the Number of Pages as Parameter}
\label{sec:number-pages-fpt}
In this section, we show that \QLE is \FPT\ when parameterized by $\kappa + \ell$, i.e., the number of missing vertices and edges, and pages.
On a high level, we use the bounded number of missing edges and pages to guess the page assignment of a hypothetical solution.
By guessing also the relative order among endpoints of new edges, we can check for a nesting relation between new edges.
However, we need to find for each new vertex a spine position that avoids nesting relations with old edges.
In \Cref{lem:two-sat}, we show that this can be done in polynomial time, which turns out to be the essential ingredient to obtain fixed-parameter tractability.
\Cref{lem:two-sat} builds on the following observation:
\begin{observation}
	\label{obs:nesting-two-edges}
	Let \ql be a queue layout of $G$ and $u_1v_1, u_2v_2 \in E(G)$ two edges of $G$ with $u_1 \prec v_1$, $u_2 \prec v_2$, and $\sigma(u_1v_1) = \sigma(u_2v_2)$. We have $u_1 \prec u_2 \Leftrightarrow v_1 \prec v_2$.
\end{observation}
\Cref{obs:nesting-two-edges} allows us to encode the (non-existence of a) nesting relation as a \probname{2-Sat} expression.
By introducing one variable $x_{u,v}$ for each pair of vertices $u,v \in V(G)$, we can express \Cref{obs:nesting-two-edges} as $x_{u_1, u_2} \Leftrightarrow x_{v_1, v_2}$.
Hence, we can describe an $\ell$-page queue layout with a fixed page assignment as a \probname{2-Sat} formula plus anti-symmetry ($x_{u,v} \Leftrightarrow \lnot x_{v,u}$) and transitivity ($x_{u,v} \land x_{v,w} \Rightarrow x_{u,w}$) constraints.
While the latter constraints require in general clauses on three literals, the total order among old and new vertices, respectively, fixes for each three vertices $u,v,w \in V(G)$ the relative order among at least two of them.
Thus, we know the truth value of $x_{u,v}$, $x_{v,w}$, or $x_{u,w}$, which allows us to also incorporate transitivity constraints in our \probname{2-Sat} formula.
As the satisfiability of a \probname{2-Sat} formula can be checked in linear time~\cite{APT.LTA.1979}, we obtain:

\begin{restatable}\restateref{lem:two-sat}{lemma}{lemmaTwoSat}
	\label{lem:two-sat}
	Given an instance $\instance = \instanceLong$ of \QLE, a page assignment~$\sigma_G$ for all edges, and a total order $\prec_{\Eadd}$ in which the endpoints of new edges will appear on the spine.
	In $\BigO{{\Size{\instance}}^2}$ time, we can check if \instance admits a solution \ql[G] where $\prec_G$ extends $\prec_{\Eadd}$ or report that no such solution exists.
\end{restatable}
\begin{prooflater}{plemmaTwoSat}
	We start the proof by carefully analysing the implications of the given restrictions.
	Regarding the page assignment~$\sigma_G$, we observe that, due to the known spine position for the endpoints of any edge $e \in \EaddH$, we can check whether an edge $e \in \EaddH$ is in a nesting relation with an old edge $e' \in E(H)$ or another new edge $e' \in \EaddH$.
	Since the spine position of their endpoints is predetermined by $\prec_{H}$, we consider them for the remainder of the proof as old edges.
	Taking also $\prec_{\Eadd}$ into account, we can check for a nesting relation between two new edges $uv, u'v' \in \Eadd$. %
	Hence, if the given assignments imply a nesting relation on the same page, which we can check in $\BigO{\madd{}\cdot \Size{\instance{}}}$ time, we can directly return that no solution exists that adheres to them.
	
	What remains to check is for a nesting relation between %
    new edges and old edges.
	However, in contrast to above, this now depends on the concrete spine position of the new vertices.
	To place each new vertex on the spine, we use a \probname{2-Sat} formula which builds on \Cref{obs:nesting-two-edges}.
	To that end, we introduce a variable $x_{u,v}$ for every two vertices $u,v \in V(G)$ with $u \neq v$.
	These variables carry the following semantic.
	\begin{align}
		\label{eq:two-sat-semantic}
		x_{u,v} = 1 \Leftrightarrow u \prec v
	\end{align}
	We use above semantics when creating the spine order from a (satisfying) variable assignment.
	To model our problem in \probname{2-Sat}, we have to add clauses that (1) capture antisymmetry, (2) guarantee that the obtained spine order respects the order among the endpoints of new edges given by $\prec_{\Eadd}$, (3) ensure that the variable assignment corresponds to a valid spine order, and (4) model \Cref{obs:nesting-two-edges}.
	The \probname{2-Sat} formula $\varphi$, that we create in the following, is of the form
	\begin{align*}
		\varphi \coloneqq  \varphi_1 \land \varphi_2 \land \varphi_3 \land \varphi_4,
	\end{align*}
	where each of the four subformulas encodes the respective part of the above description.\looseness=-1
	
	\subparagraph*{Clauses.}
	We now describe each of the four subformulas in more detail.
	
	The first subformula $\varphi_1$ should encode antisymmetry, i.e., should prevent that we have $x_{u,v}$ and $x_{v,u}$ or $\lnot x_{u,v}$ and $\lnot x_{v,u}$ for $u,v \in V(G)$.
	This can be expressed as:
	\begin{align*}
			\varphi_1 \coloneqq\bigwedge_{\substack{u,v \in V(G)\\u \neq v}} ((x_{u,v} \Rightarrow \lnot x_{v,u}) \land ((x_{v,u} \Rightarrow \lnot x_{u,v}))
	\end{align*}
		
	The purpose of the second subformula $\varphi_2$ is to guarantee that the spine order $\prec_{G}$ obtained from the variable assignment is an extension of $\prec_{\Eadd}$ w.r.t.\ the new vertices.
	Let $u,v\in \Vadd, w \in V(G) \setminus \Vadd$ be three vertices with $u \prec_{\Eadd} v$.
	From $u \prec_{\Eadd} v$, we conclude that $v \prec w$ implies $u \prec w$ as otherwise $v \prec_{G} u$ would need to hold.
	We can express this in \probname{2-Sat} as $x_{v,w} \Rightarrow x_{u,w}$:
	\begin{align*}
		\varphi_2 \coloneqq \bigwedge_{\substack{u,v \in \Vadd,\\u \prec_{\Eadd} v}} \:\bigwedge_{w \in V(G) \setminus \Vadd} x_{v,w} \Rightarrow x_{u,w}
	\end{align*}
	
	The next subformula $\varphi_3$ ensures an analogous concept with respect to the spine position of old vertices.
	To that end, we need to ensure that for two vertices $v,w \in V(H)$ with $v \prec_H w$ and for any spine order $\prec_G$ with $u \prec_G v$ we also have $u \prec_G w$.
	Note that otherwise we would need to have $w \prec_G u \prec_G v$ which would contradict $v \prec_H w$.
	Similar to above, we can express this condition as $x_{u,v} \Rightarrow x_{u,w}$:
	\begin{align*}
		\varphi_3 \coloneqq \bigwedge_{u \in \Vadd}\:\bigwedge_{\substack{v,w \in V(H),\\v \prec_H w}} x_{u,v} \Rightarrow x_{u,w}
	\end{align*}
	
	Finally, the last subformula, subformula $\varphi_4$, should ensure that the obtained spine order $\prec_G$ together with the given page assignment $\sigma_G$ yields a queue layout \ql[G] of $G$.
	To that end, we formalize \Cref{obs:nesting-two-edges}.
	Let $u_1v_1, u_2v_2 \in E(G)$ be a pair of edges for which we need to encode \Cref{obs:nesting-two-edges} and where, in addition, one of the edges is old.
	Note that otherwise we can already determine if the edges are in a nesting relation.
	To express \Cref{obs:nesting-two-edges} as \probname{2-Sat} formula, we use $x_{u_1,u_2} \Leftrightarrow x_{v_1,v_2}$ and replace the bi-equivalence with a conjunction of two implications.
    Note that from $\prec_{\Eadd}$ and $\prec_H$ we can deduce the relative order of the endpoints of all edges of $G$. 
	To define $\varphi_4$, we let $E' \subseteq E(G) \times E(G)$ denote the pairs of edges for which we need to encode \Cref{obs:nesting-two-edges} and set
	\begin{align*}
		\varphi_4 \coloneqq \bigwedge_{u_1v_1, u_2v_2 \in E'} x_{u_1,u_2} \Leftrightarrow x_{v_1,v_2}.
	\end{align*}

	\subparagraph*{Putting all Together.}
	The number of variables in the above-introduced formula $\varphi$ is in \BigO{\Size{\instance}^2}.
	Furthermore, the number of clauses is polynomial in the size of \instance.
	Note that $\varphi$ contains a variable for every pair of vertices.
	However, the order among two old vertices and new vertices is predetermined by $\prec_H$ and $\prec_{\Eadd}$, respectively.
	Hence, we replace for every two old vertices $u, v \in V(H)$ in every clause of~$\varphi$ the variable~$x_{u,v}$ with $1$ if $u \prec_H v$; otherwise we replace $x_{u, v}$ by $0$.
	We do so similarly for every two vertices $u, v \in \Eadd$
    Note that by fixing the variables for edges in $\Eadd$ w.r.t.\ $\prec_{\Eadd}$, we also ensure that our obtained solution will extend $\prec_{\Eadd}$.
	Furthermore, we can decide if $\varphi$ has a satisfying assignment in \BigO{\Size{\instance}^2} time~\cite{APT.LTA.1979}.
	What remains to show is that a satisfying assignment of $\varphi$ corresponds to a spine order $\prec_G$ that, together with $\sigma_G$, yields a solution \ql and vice versa.
	
	For the $(\Rightarrow)$-direction, let $\Psi$ be a variable assignment that satisfies $\varphi$.
	We construct based on $\Psi$ a spine order $\prec_G$ of $G$.
	Recall that we replaced some variables with the truth value predetermined by $\prec_H$ or $\prec_{\Eadd}$, i.e., we cannot directly read $\prec_G$ off from $\Psi$.
	For every two old vertices $u,v \in V(H)$ with $u \neq v$, we set $u \prec_G$ if and only if $u \prec_H v$.
	We perform similarly for every vertices constrained by $\prec_{\Eadd}$.
	This ensures that $\prec_G$ extends $\prec_{H}$ and $\prec_{\Eadd}$.
	To fix the order among a new and an old vertex, we consider for every new vertex $u$ the variables $x_{u,v}$ and $x_{v, u}$ with $v \in V(H)$.
	We set $u \prec_G v$ if $\Psi(x_{u,v}) = 1$ and $v \prec_G u$ if $\Psi(x_{v,u}) = 1$.
	Thanks to subformula $\varphi_1$, which ensures antisymmetry, exactly one of these cases is true.
	Furthermore, thanks to subformulas $\varphi_2$ and $\varphi_3$ the order $\prec_G$ is well-defined: $w \prec_G u$ and $v \prec_G w$ but $u \prec_H v$ or $u \prec_{\Eadd} v$ is not possible.
	To see that $\prec_G$ defines a partial order, consider three vertices $u,v,w  \in V(G)$ with $u \prec_G v$ and $v \prec_G w$ for which then $u \prec_G w$ must hold.
	Expressed as formula, this corresponds to $x_{u,v} \land x_{v,w} \Rightarrow x_{u,w}$, which is equivalent to $\lnot x_{u,v} \lor \lnot x_{v,w} \lor x_{u,w}$.
	We observe that independent of the choice of $u$, $v$, and $w$, the order of at least two of them is prescribed by $\prec_H$ or $\prec_{\Eadd}$.
	Hence, above expression is either vacuously true or can be simplified to a clause on two literals corresponding to a clause from $\varphi_1$ or $\varphi_2$.
	Consequently, we can take the transitive closure of $\prec_G$ to obtain the final spine order.
	To complete the construction of \ql[G], we take the given page assignment.
	Due to $\varphi_4$, which is a formalization of \Cref{obs:nesting-two-edges}, we conclude that \ql[H] is a solution of \instance that extends \ql[H] and $\prec_{\Eadd}$.
	
	For the $(\Leftarrow)$-direction, let \ql be a solution of \instance where \ql extends \ql[H] and $\prec_G$ also extends $\prec_{\Eadd}$.
	We now construct a variable assignment $\Psi$ from $\prec_G$ by setting $\Psi(x_{u,v}) = 1$ if and only if $u \prec_G v$.
	We now show that $\Psi$ satisfies all formulas $\varphi_i$ for $i \in [4]$, thus satisfying $\varphi$.
	As $u \prec_G v$ is a total order of $V(G)$, it is antisymmetric and thus is $\varphi_1$ satisfied by $\Psi$.
	Furthermore, since $\prec_{G}$ extends $\prec_{\Eadd}$ and $\prec_{H}$ the subformulas $\varphi_2$ and $\varphi_3$ are satisfied by construction, respectively.
	Finally, as \ql extends \ql[H] no two edges of $G$ on the same page are in a nesting relation under $\prec_G$ given $\sigma_G$.
	Hence, since $\prec_G$ avoids $u_1 \prec_G u_2 \prec_G v_2 \prec_G v_1$ for any two edges $u_1v_1, u_2v_2 \in E(G)$ with $u_1 \prec_G v_1$, $u_2 \prec_G v_2$, and $u_1 \prec_G u_2$, also subformula $\varphi_4$ is satisfied under $\Psi$
	Combining all, we conclude that $\Psi$ satisfies~$\varphi$.
	
	Having now established both directions, we can conclude that \instance admits a solution that extends $\prec_{\Eadd}$ if and only if $\varphi$ is satisfiable.
\end{prooflater}
Finally, we observe that there are $\BigO{\ell^{\madd}}$ different page assignments $\sigma_G$ and \BigO{{\nadd{}}!\cdot {\madd}^{\nadd}} potential orders among endpoints of new edges.
Since we apply \Cref{lem:two-sat} in each of the $\BigO{\ell^{\madd} \cdot {\nadd{}}!\cdot {\madd}^{\nadd}}$ different branches, we %
obtain the following theorem.
\thmkappapagesfpt*

\section{\textsc{QLE} With Two Missing Vertices}
\label{sec:two-missing-vertices}
In this section, we show that \QLE is polynomial-time solvable for the case where the missing part consists of only two vertices of arbitrary degree.
Recall that the problem of extending stack layouts is already \NP-hard in this setting~\cite{DFGN.PCE.2024}.

Let $\instance$ be an instance of \QLE with $\Vadd = \{u,v\}$.
First, we branch to determine the placement of $u$ and $v$ on the spine. 
In each of the $\BigO{\Size{\instance}^2}$ branches, the spine order $\prec_G$ is fixed, and we only need to assign the new edges to pages.
However, in contrast to \Cref{thm:kappa-xp}, the number of new edges is unbounded.
Nevertheless, we can avoid the intractability result of \Cref{thm:only-edges-hard} by using the fact that all missing edges are incident to $u$ or $v$, which allows us to
observe useful properties of the visibility relation on certain pages. %

For the remainder of this section, we assume a fixed spine order $\prec_G$ of $G$ where, w.l.o.g., $u \prec_G v$.
Furthermore, we can branch to determine the page assignment $\sigma(uv)$ of $uv$, if it exists, and check in linear time whether $uv$ is in a nesting relation with an old edge on $\sigma(uv)$.
Hence, we consider the edge $uv$ for the upcoming discussion as an old edge.

We define for every edge $e \in \Eadd$ the sets $P(e)$ and $\nestingRelation(e)$ of \emph{admissible} pages and \emph{conflicting} new edges, respectively.
\ifthenelse{\boolean{long}}{Recall that}{More formally,} $P(e)$ contains all pages $p \in [\ell]$ such that $\langle \prec_G, \sigma'_H \rangle$ is a queue layout of the graph $G'$ with $V(G') = V(G)$ and $E(G') = E(H) \cup \{e\}$ where $\sigma'_H(e) = p$ and $\sigma'_H(e') = \sigma_H(e')$ for all $e' \in E(H)$.
Furthermore, $\nestingRelation(e)$ contains all new edges $e' \in \Eadd$ that are in a nesting relation with $e$ w.r.t.\ $\prec_G$.
\ifthenelse{\boolean{long}}{Note that we can compute for a single edge $e \in \Eadd$ both sets in linear time. We first make some simple observation that allows us to either immediately conclude that \instance is a negative instance or fix parts of the page assignment $\sigma_G$.}{We now make a simple observation about the relation among the sets $P(e)$ and $\nestingRelation(e)$ and the existence of a solution and its page assignment.}
\begin{observation}
	\label{obs:two-missing-vertices-simple-cases}
	Let $e \in \Eadd$ be a new edge.
	If $P(e) = \emptyset$, then \instance is a negative instance.
	Furthermore, if \instance is a positive instance and there exists a page $p \in P(e) \setminus \bigcup_{e' \in \nestingRelation(e)} P(e')$, then there exists a solution $\langle \prec_G, \sigma\rangle$ with $\sigma(e) = p$. 
\end{observation}
\Cref{obs:two-missing-vertices-simple-cases} allows us to restrict our attention to parts of the instance where for every missing edge $e$ each admissible page $p$ is also admissible for some other new edge $e' \in \nestingRelation(e)$.
Note that this also implies $\nestingRelation(e) \neq \emptyset$.
Next, we show a relation among the admissible pages for different new edges\ifthenelse{\boolean{long}}{ that turns out to be crucial for obtaining our efficient algorithm}{}. 
\begin{restatable}\restateref{lem:propage-admissible-pages}{lemma}{lemmaPropagateAdmissiblePages}
\label{lem:propage-admissible-pages}
Let $e_1 = vx$, $e_2 = uy$, and $e_3 = uz$ be three new edges such that $u \prec x \prec y \prec z$ and $v \prec y$.
For any page $p \in P(e_1) \cap P(e_3)$ it holds $p \in P(e_2)$.
\end{restatable}
\begin{proofsketch}
For the sake of a contradiction, assume $p \notin P(e_2)$.
This means there is an old edge $e = ab$ on $p$ that is in a nesting relation with $e_2$.
If $e$ is nested by $e_2$, it is also nested by $e_3$; see \Cref{fig:admissible-pages-propagation}a.
Thus, we have $p \notin P(e_3)$, which is a contradiction.
Similarly, if $e$ nests $e_2$, then it must also nest $e_1$; see \Cref{fig:admissible-pages-propagation}b.
As this implies $p \notin P(e_1)$, we again arrive at a contradiction.
As both cases lead to a contradiction, we conclude that $p \in P(e_2)$ must hold.
\end{proofsketch}
\begin{prooflater}{plemmaPropagateAdmissiblePages}
Towards a contradiction, assume that there exists a page $p \in P(e_1) \cap P(e_3)$ with $p \notin P(e_2)$.
This means there exists an edge $e = ab$ with $\sigma(e) = p$ that is in a nesting relation with $e_2$.
Assume without loss of generality $a \prec b$.
There are two cases to consider also visualized in \Cref{fig:admissible-pages-propagation}: The edge $e$ could be nested by $e_2$ or nests $e_2$.
The former case implies $u \prec a \prec b \prec y$; see also \Cref{fig:admissible-pages-propagation}a.
Since we have $y \prec z$, we have $u \prec a \prec b \prec z$, which implies that $e$ is also nested by $e_3 = uz$.
However, this contradicts $p \in P(e_1) \cap P(e_3)$ as we have $p \notin P(e_3)$ by definition of $P(\cdot)$.
For the latter case, we get $a \prec u \prec y \prec b$; see also \Cref{fig:admissible-pages-propagation}b. 
Together with $v \prec y$ and $u \prec x \prec y$, we derive $a \prec x \prec b$ and $a \prec v \prec b$, which implies that $e$ also nests $e_1$.
However, this means $p \notin P(e_1)$ by the definition of $P(\cdot)$, contradicting $p \in P(e_1) \cap P(e_3)$.
Since both cases lead to a contradiction, we conclude $p \in P(e_2)$ must hold.
Since $p$ was selected arbitrarily, this holds for all $p \in P(e_1) \cap P(e_3)$.
\end{prooflater}
\begin{figure}
	\centering
	\includegraphics[page=1]{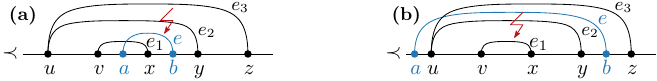}
	\caption{If the page $p$ is admissible for $e_1$ and $e_3$, it is also admissible for $e_2$: No edge $e$ (on page $p$) can be in a nesting relation with only $e_2$.}
	\label{fig:admissible-pages-propagation}
\end{figure}

\ifthenelse{\boolean{long}}{
We can state a property that is symmetric to \Cref{lem:propage-admissible-pages}.
To show it, we can use statements that are analogous to those for \Cref{lem:propage-admissible-pages}.
}{}
\begin{statelater}{lemmaPropageAdmissiblePagesLeft}
    \begin{lemma}
    \label{lem:propage-admissible-pages-left}
    Let $e_1 = ux$, $e_2 = vy$, and $e_3 = vz$ be three new edges such that $z \prec y \prec x \prec v$ and $y \prec u$.
    For any page $p \in P(e_1) \cap P(e_3)$ it holds $p \in P(e_2)$.
    \end{lemma}
\end{statelater}

A solution \emph{uses} a spine order $\prec^*_G$ if we have $\langle \prec^*_G, \sigma_G \rangle$.
We now use \Cref{lem:propage-admissible-pages} to show that for two missing vertices with a fixed placement on the spine we can safely remove some edges with at least two admissible pages:
\begin{restatable}\restateref{lem:two-missing-vertices-remove-safe}{lemma}{lemmaTwoMissingVerticesRemoveSafe}
    \label{lem:two-missing-vertices-remove-safe}
    Let $\instance = \instanceLong$ be an instance of \QLE with $\Vadd = \{u,v\}$ and let $\prec_H\ \subseteq\ \prec_G$ be a spine order with $u \prec_G v$ that contains an edge $e = vx \in \Eadd$ with $u \prec_G x$ and $\Size{P(e)} \geq 2$.
    \instance admits a solution that uses $\prec_G$ if and only if $\instance' = \left(\ell, G' =  G \setminus \{e\}, H, \ql[H]\right)$ admits one that uses $\prec_G$.
\end{restatable}
\begin{proofsketch}
Since deleting an edge can never invalidate an existing solution, we focus on the more involved direction of showing that we can re-insert $e$ into a solution $\langle \prec_G, \sigma_{G'} \rangle$ to $\instance'$.
As the spine order $\prec_G$ is fixed, we can list all new edges $\nestingRelation(e) = \{uw_1, uw_2, \ldots, uw_k\}$ that are in a nesting relation with $e$.
Note that they must be incident to $u$.
Thus, we can sort them by their single old endpoint, i.e., $w_1 \prec_G \ldots \prec_G w_k$.
If there exists a page $p \in P(e)$ such that for all $e' \in \nestingRelation(e)$ we have $\sigma_{G'}(e') \neq p$, we can set $\sigma_G(e) = p$.
Otherwise, we process the $w_i$, $i \in k$ in decreasing order from $i = k - 1$ to $i = 1$.
If $\sigma_{G'}(e_{i + 1}) \neq \sigma_{G'}(e_{i})$, we change the page assignment of $e_{i}$ to match the one from $e_{i + 1}$.
By \Cref{lem:propage-admissible-pages}, we know that $e_i$ cannot be in a nesting relation with an old edge on page $\sigma_{G'}(e_{i + 1})$.
Furthermore,~$e_i$ cannot be in a nesting relation with a new edge~$e'$ on page $\sigma_{G'}(e_{i + 1})$ as then $\sigma_{G'}(e_{i + 1})$ would also be in a nesting relation with $e'$ as \Cref{fig:reduction-rule}a underlines.
Thus, setting $\sigma_{G'}(e_{i}) = \sigma_{G'}(e_{i + 1})$ is safe.
Once we have eventually processed the last edge $e_1$, there exists at least one page $p \in P(e)$ such that no edge $e' \in \nestingRelation(e)$ is still assigned to $p$.
We can now safely set $\sigma_G(e) = p$.
\end{proofsketch}
\begin{prooflater}{plemmaTwoMissingVerticesRemoveSafe}
For the $(\Rightarrow)$-direction, it is sufficient to observe that we can take any solution \ql[G] (that uses $\prec_G$) and remove the edge $e$ from the page assignment $\sigma_G$.
The obtained page assignment $\sigma_{G'}$ witnesses the existence of a solution $\langle \prec_G, \sigma_{G'}\rangle$ for $\instance'$ that uses $\prec_G$.
Hence, we focus for the remainder of the proof on the more involved $(\Leftarrow)$-direction.

To that end, let $\langle \prec_G, \sigma_{G'} \rangle$ be a solution for $\instance'$ that uses $\prec_G$.
To obtain a solution \ql[G] for $G$ we only need to assign the edge $e = vx$ to a page $p \in [\ell]$.
Observe that the spine order of $G$ is fixed, i.e., the end points of $e$ have a fixed position on the spine.
Clearly, if there exists a page $p \in P(e)$ such that $e$ is in no nesting relation with a new edge incident to $u$, we can assign $e$ to the page $p$ and are done.
Hence, let us consider the case where $e$ is on every page $p \in P(e)$ in a nesting relation with another edge, also depicted in \Cref{fig:reduction-rule}a.
Let these edges be $e_1, \ldots, e_k$ with $k \geq \Size{P(e)}$, i.e., $\sigma_{G'}(e_i) \in P(e)$ for all $i \in [k]$ and $\bigcup_{i \in [k]} \sigma_{G'}(e_i) = P(e)$.
Observe that due to the definition of $P(e)$, each $e_i$ for $i \in [k]$ is a new edge incident to $u$, i.e., we have $e_i = uw_i$ for $w_i \in V(H)$.
Let the edges $e_1, \ldots, e_k$ be ordered by their endpoint $w_i$, i.e., we have $w_1 \prec_G \ldots \prec_G w_k$.
Recall that $e = vx$ and observe $u \prec_G w_i$, $x \prec_G w_i$, and $v \prec_G w_i$ for every $i \in [k]$. 
We now consider each edge $e_i$ from $i = k - 1$ to $i = 1$ in this order and define $\sigma^{k}_{G'} = \sigma_{G'}$.
For every $i \in [ k - 1]$, if $\sigma^{i + 1}_{G'}(e_i) = \sigma^{i + 1}_{G'}(e_{i + 1})$, we set $\sigma^{i}_{G'} = \sigma^{i + 1}_{G'}$ and proceed to $e_{i - 1}$.
Otherwise, i.e., if $\sigma^{i + 1}_{G'}(e_i) \neq \sigma^{i + 1}_{G'}(e_{i + 1})$, we consider the page $p$ of $e_{i + 1}$, i.e., the page $p \in [\ell]$ with $\sigma^{i + 1}_{G'}(e_{i + 1}) = p$.
We now show that we can re-assign $e_i$ to page $p$, i.e., that $\langle \prec_G, \sigma^{i}_{G'} \rangle$ is also a solution to $\instance'$, where $\sigma^{i}_{G'}$ differs from $\sigma^{i + 1}_{G'}$ only in $\sigma^{i}_{G'}(e_i) = p$.
To that end, we first observe $u \prec_G x \prec_G w_i \prec_G w_{i + 1}$ and $v \prec_G w_i$.
As $\sigma^{i + 1}_{G'}(e_{i + 1}) \in P(e)$ and $\sigma^{i + 1}_{G'}(e_{i + 1}) = p$, we conclude $p \in P(e)$ and $p \in P(e_{i + 1})$.
Thus, using \Cref{lem:propage-admissible-pages}, we conclude $p \in P(e_i)$.
By the definition of $P(e_i)$, this means that $e_i$ cannot be in a nesting relation with an old edge on $p$.
Hence, the only reason that could hinder us from assigning $e_i$ to page $p$ is a new edge on $p$ that is in a nesting relation with $e_i$.
Let this new edge be $e'$, see also the dashed edge in \Cref{fig:reduction-rule}a, and observe that $e'$ must be incident to $v$, i.e., $e' = vz$ with $z \in V(H)$.
As $u \prec_G v$, and $e_i$ and $e'$ are in a nesting relation, we must have $z \prec_G w_i$, i.e., $e_i$ nests $e'$.
However, $w_i \prec_G w_{i + 1}$ implies that also $e_{i + 1}$ nests $e'$.
However, this is a contradiction to the existence of the solution $\langle \prec_G, \sigma^{i + 1}_{G'}\rangle$ as $\sigma^{i + 1}_{G'}(e_i) = \sigma^{i + 1}_{G'}(e') = p$.
Hence, the edge $e'$ cannot exist, and we can safely assign $e_i$ to page $p$ and obtain a valid solution $\langle \prec_G, \sigma^{i}_{G'} \rangle$ of $\instance'$.
We continue this process with $e_{i - 1}$ until we have eventually handled the edge $e_1$.
By above arguments, $\langle \prec_G, \sigma^1_{G'} \rangle$ is a valid solution to $\instance'$.
Recall that we had $\bigcup_{i \in [k]} \sigma_{G'}(e_i) = P(e)$.
Since $\sigma_{G'}(e_i) \in P(e)$ for all $i \in [k]$, there must exist one $j \in [k]$ such that $\sigma^1_{G'}(e_j) \neq \sigma_{G'}(e_j) = p^*$.
Thanks to the above procedure, for every new edge $e'$ that nests $e$ we have $\sigma_{G'}^1(e') \neq p^*$.
Thus, we conclude that $\langle \prec_G, \sigma^1_{G'} \rangle$ extended by $\sigma^1_{G'}(e) = p^*)$ is a solution to $\instance$.
\end{prooflater}
\begin{figure}
	\centering
	\includegraphics[page=1]{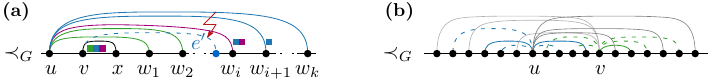}
	\caption{\textbf{(a)} Visualization of the reduction rule from \Cref{lem:two-missing-vertices-remove-safe}. Colors indicate the page assignment and squares the admissible pages for some of new edges. The dashed edge~$e'$ cannot exist. \textbf{(b)}~The obtained instance after applying the reduction rule. Solid, colored edges have only one admissible page, dashed, colored edges got removed, and the page assignment of the gray edges depends on the one of the solid, colored edges. These edges are thus retained.}
	\label{fig:reduction-rule}
\end{figure}
\begin{statelater}{lemmaTwoMissingVerticesRemoveSafeLeft}
    \Cref{lem:propage-admissible-pages-left} allows us to define a reduction rule that removes new edges that are incident to $u$ and have their endpoint left of $v$.
    The following lemma formally states this rule, which is symmetric to \Cref{lem:two-missing-vertices-remove-safe}.
    To show it, we can use almost the same proof strategy as for \Cref{lem:two-missing-vertices-remove-safe}.
    However, we now use \Cref{lem:propage-admissible-pages-left} to deduce that that $\sigma_{G'}(e_{i + 1}) \in P(e_i)$ must hold.
    \begin{lemma}
        \label{lem:two-missing-vertices-remove-safe-left}
        Let $\instance = \instanceLong$ be an instance of \QLE with $\Vadd = \{u,v\}$ and let $\prec_H\ \subseteq\ \prec_G$ be a spine order with $u \prec_G v$ that contains an edge $e = ux \in \Eadd$ with $x \prec_G v$ and $\Size{P(e)} \geq 2$.
        \instance admits a solution that uses $\prec_G$ if and only if $\instance' = \left(\ell, H, G' =  G \setminus \{e\}, \ql\right)$ admits one that uses $\prec_G$.
    \end{lemma}
\end{statelater}

We now have all tools at hand to show our final result:

\theoremTwoMissingVerticesPoly*
\begin{proofsketch}
We branch to determine the spine order $\prec_G$ and the page assignment for the edge $uv$.
For each of the $\BigO{\Size{\instance}^2 \cdot \ell}$ branches, we check if the edge $uv$ is in a nesting relation with an old edge.
If not, we compute the sets $P(\cdot)$ and $\nestingRelation(\cdot)$.
Using \Cref{obs:two-missing-vertices-simple-cases}, we remove edges $e$ with $\left(P(e) \setminus \bigcup_{e' \in \nestingRelation(e)} P(e')\right) \neq \emptyset$.
{%
Furthermore, we remove all edges $vx \in \Eadd$ with $u \prec_G x$ and $\Size{P(vx)} \geq 2$.
\Cref{lem:two-missing-vertices-remove-safe} ensures that this operation is safe.
Symmetrically, we remove all edges $uy \in \Eadd$ with $y \prec_G v$ and $\Size{P(uy)} \geq 2$.
Safeness is guaranteed by a reduction rule that is symmetric to \Cref{lem:two-missing-vertices-remove-safe}.
We provide the formal statement in \Cref{lem:two-missing-vertices-remove-safe-left} in \Cref{app:omitted-lemmas}.
}
Each new edge $e$ now either has $\Size{P(e)} = 1$ or is of the form $ux'$ with $v \prec_G x'$ or $vy'$ with $y' \prec u$; see \Cref{fig:reduction-rule}b for an example.
Observe that these edges $ux'$ and $vy'$ are not in a nesting relation.
We assign edges with a single admissible page to the respective page.
Afterwards, we assign the remaining new edges $e$ to a page $\sigma(e) = p$ such that for all new edges $e' \in \nestingRelation(e)$ with $\Size{P(e')} = 1$ we have $\sigma(e') \neq p$.
If this is not possible, the instance has no solution due to \Cref{lem:two-missing-vertices-remove-safe,lem:two-missing-vertices-remove-safe-left}.
Otherwise, the same lemmas in concert with \Cref{obs:two-missing-vertices-simple-cases} allow us to complete the obtained page assignment to a solution.
Overall, this takes $\BigO{\madd \cdot \Size{\instance}}$ time per branch.
\end{proofsketch}
\begin{prooflater}{ptheoremTwoMissingVerticesPoly}
Let \instance be a instance of \QLE with $\nadd = 2$ and let $\Vadd = \{u,v\}$ the two missing vertices.
First, we branch to determine the final spine order $\prec_G$ and the page to which the edge $uv$, if it exists, should be assigned to.
This gives us $\BigO{\Size{\instance}^2 \cdot\ell}$ different branches, each corresponding to a different combination of spine order $\prec_G$ and page assignment $\sigma_G(uv)$.
For the sake of presentation, we assume without loss of generality $u \prec_G v$.
For each branch, we can check in linear time if the edge $uv$ is in a nesting relation with an old edge on $\sigma_G(uv)$.
If this is the case, we discard the branch.
Otherwise, we also compute for each edge $\Eadd \setminus \{uv\}$ in linear time the sets $P(e)$ and $\nestingRelation(e)$.
We assume for the remainder of this proof that no new edge $e$ has $P(e) = \emptyset$.
Otherwise, we can immediately return that \instance does not have a solution thanks to \Cref{obs:two-missing-vertices-simple-cases}.
In addition, \Cref{obs:two-missing-vertices-simple-cases} allows us to remove all new edges $e$ with $P(e) \setminus \bigcup_{e' \in \nestingRelation(e)} P(e') \neq \emptyset$.
With the aim to apply \Cref{lem:two-missing-vertices-remove-safe,lem:two-missing-vertices-remove-safe-left}, we remove all edges $vx$ and $uy$ with $u \prec_G x$, $y \prec_G v$, and $\Size{P(vx)}, \Size{P(uy)} \geq 2$.
Overall, this takes $\BigO{\madd \Size{\instance}}$ time.

In the resulting instance $\instance'$, each new edge $e$ either has $\Size{P(e)} = 1$ or is of the form $ux'$ with $v \prec_G x'$ or $vy'$ with $y' \prec u$; see \Cref{fig:reduction-rule}b for an example.
For new edges with a single admissible page, we have no choice but must assign these edges to their respective unique admissible page.
For the other remaining new edges $e$, we check whether we can still assign each of them to one of their admissible pages, i.e., whether there is a page $p \in P(e)$ such that for all new edges $e' \in \nestingRelation(e)$ with $\Size{P(e')} = 1$ we have $\sigma(e') \neq p$.
If there exists a new edge for which this is not the case, we conclude, by \Cref{lem:two-missing-vertices-remove-safe,lem:two-missing-vertices-remove-safe-left}, that the instance does not admit a solution.
Otherwise, it is sufficient to observe that no two edges $ux'$ and $vy'$ with $v \prec_G x'$ and $y' \prec u$ are in a nesting relation, as we have $y' \prec_G u \prec_G v \prec x'$.
Hence, we have obtained a solution $\langle \prec_G, \sigma_{G'}\rangle$ to $\instance'$.
Iteratively applying first \Cref{lem:two-missing-vertices-remove-safe,lem:two-missing-vertices-remove-safe-left} and then \Cref{obs:two-missing-vertices-simple-cases}, we can complete $\langle \prec_G, \sigma_{G'}\rangle$ to a solution \ql[G] of \instance.
This takes \BigO{\madd\cdot \Size{\instance}} time.
Overall, the running time can be bounded by $\BigO{\Size{\instance}^2 \cdot \ell \cdot (\madd \cdot \Size{\instance})}$.
\end{prooflater}

\section{Concluding Remarks}
\label{sec:conclusion}
This paper provides a detailed analysis of the (parameterized) complexity of the problem of extending queue layouts.
While we can draw for some results parallels to the related extension problem for stack layouts, our investigation also uncovers surprising differences between the two problems.

We see generalizing \Cref{thm:two-missing-vertices-poly} to $k$ (high-degree) missing vertices as the main challenge for future work.
While the equivalence between queue layouts and colorings of permutation graphs immediately yields membership in \XP\ when parameterized by \nadd and $\ell$, \Cref{thm:two-missing-vertices-poly} suggests that we can maintain membership in \XP\ even after dropping the number of pages $\ell$ from our parametrization.
{%
Recall that we required $H$ to be a subgraph of $G$.
A generalization of the studied problem relaxes this requirement, allowing the spine order for some vertices and the page assignment for some edges of $G$ to be specified as part of the input.
We see this as an interesting direction for future investigations.
}%
{%
Finally, a natural next step is to extend our results to other linear layout drawing styles.
}
 \bibliography{references}

\newpage
\appendix
\ifthenelse{\boolean{long}}{}{
\section{Omitted Proofs}
\label{app:omitted-proofs}

\subsection{Omitted Proofs from Section~4}

\theoremOnlyEdgesFPT*
\label{thm:only-edges-fpt*}
Before we can show \Cref{thm:only-edges-fpt}, we first show that we can remove all new edges from the instance that can be placed in at least $\madd$ pages without being in a nesting relation with an old edge on that page.
Let $P(e) \subseteq [\ell]$ be the set of \emph{admissible pages} for a new edge $e \in \Eadd$, i.e., the set of pages such that $p \in P(e)$ if and only if $\langle\prec_{H},\sigma'_{H}\rangle$ is an $\ell$-page queue layout of $H\cup \{e\}$, where $\sigma'_{H}(e) = p$ and $\sigma'_{H}(e') = \sigma_{H}(e')$ for all other edges $e' \in E(H)$.
\begin{lemma}
\stateLemmaOnlyEdgesRemoveSafeText
\end{lemma}
\plemmaOnlyEdgesRemoveSafe
With \Cref{lem:only-edges-remove-safe} at hand, we are now ready to show \Cref{thm:only-edges-fpt}.
\ptheoremOnlyEdgesFPT

\subsection{Omitted Proofs from Section~5}
\theoremKappaXP*
\label{thm:kappa-xp*}
\ptheoremKappaXP

\theoremKappaWHard*
\label{thm:kappa-w1*}
\ptheoremKappaWHard

\subsection{Omitted Proofs from Section~6}
\lemmaTwoSat*
\label{lem:two-sat*}
\plemmaTwoSat

\subsection{Omitted Proofs from Section~7}
\lemmaPropagateAdmissiblePages*
\label{lem:propage-admissible-pages*}
\plemmaPropagateAdmissiblePages

\lemmaTwoMissingVerticesRemoveSafe*
\label{lem:two-missing-vertices-remove-safe*}
\plemmaTwoMissingVerticesRemoveSafe

\theoremTwoMissingVerticesPoly*
\label{thm:two-missing-vertices-poly*}
\ptheoremTwoMissingVerticesPoly

\newpage

\section{On Permutation Graph Colorings and Fixed-Order Queue Layouts}\label{app:permutation}
Dujmović and Wood already note that ``the problem of assigning edges to queues in a fixed vertex ordering is equivalent to colouring a permutation graph'' \cite{DW.LLG.2004} and attribute this result to a 1941 publication by Dushnik and Miller \cite{DM.POS.1941}.
To make this insight more formal and also accessible in modern terms, we want to independently show this relationship in this section.
We want to thank the participants of the HOMONOLO 2024 workshop who helped us uncover the relationship to permutation graphs.

\permutationGraphs

Recall that \NP-completeness of the precoloring extension problem for permutation graphs~\cite{Jan.OCC.1997} has been shown as well as \XP-membership for list-coloring when parameterized by the number $k$ of colors~\cite{EST.LCL.2014}.
Thus, this shows the following theorem.
\thmVGisVH*
\label{thm:only-edges-hard*}

\newpage
\section{Omitted Details from the \W[1]-hardness Reduction}
\label{app:reduction}
\subsection{Edge Gadget}
\detailsEdgeGadget

\subsection{Fixation Gadget}
\detailsFixationGadget

\newpage

\section{Omitted Lemma Statements From \Cref{sec:two-missing-vertices}}
\label{app:omitted-lemmas}
In this section, we provide the lemma statements that we omitted from the main text due to space constraints.

The following lemma states a symmetric property to \Cref{lem:propage-admissible-pages} w.r.t.\ new edges that have $u$ as one endpoint and their second endpoint left of $v$.
Its proof is analogs to the one of \Cref{lem:propage-admissible-pages}.
\lemmaPropageAdmissiblePagesLeft

\lemmaTwoMissingVerticesRemoveSafeLeft

\newpage

}
\section{Removing Multi-Edges in the \W[1]-hardness Reduction}
\label{app:remove-multi-edges}
In this section, we describe how we can adapt the reduction used to obtain \W[1]-hardness in \Cref{thm:kappa-w1} to not rely on multi-edges.
To that end, recall that we introduce multi-edges in two places: First in the edge gadget, since we create, e.g., the edge $u_{\alpha}^{\bot_L}u_{\alpha}^{i + 1}$ in every gadget for an edge $e \in E(G_C)$ with endpoint $u_{\alpha}^i$.
Second in the fixation gadget, where we create, e.g., the edge $u_1^{\bot_L}u_1^{\bot}$ on every page $p \in [\ell] \setminus \{p_d\}$.
Our strategy to remove these multi-edges is to introduce multiple copies of certain vertices and distribute the multi-edges among them.

More concretely, we introduce for every edge $e\in E(G_C)$ and every color $\alpha \in [k + 1]$ the vertices $u_{\alpha}^{\bot^e_L}$ and $u_{\alpha}^{\bot^e_R}$.
These vertices will take the role of $u_{\alpha}^{\bot_L}$ and $u_{\alpha}^{\bot_R}$, and we remove, therefore, the latter vertices from our construction.
For the remainder of this section, we assume an (arbitrary) order $e_1, e_2, \ldots, e_{M}$ of the $M = \Size{E(G_C)}$-many edges of $G_C$.
We order the additional vertices based on their index, i.e., we set $u_{\alpha}^{\bot^{e_i}_L} \prec u_{\alpha}^{\bot^{e_{i + 1}}_L}$ and $u_{\alpha}^{\bot^{e_i}_R} \prec u_{\alpha}^{\bot^{e_{i + 1}}_R}$ for every $i \in [M - 1]$ and $\alpha \in [k + 1]$.
Furthermore, we set $u_{\alpha}^{\bot^{e_M}_L} \prec u_{\alpha}^{\bot} \prec u_{\alpha}^{\bot^{e_1}_R}$ and, for every $\alpha \in [k]$, $u_{\alpha}^{\bot^{e_M}_R} \prec u_{\alpha}^{1}$ and $u_{\alpha}^{n_{\alpha} + 1} \prec u_{\alpha + 1}^{\bot^{e_1}_L}$.
The remainder of the spine order is identical to the description in \Cref{sec:kappa-w-1} and \Cref{fig:w-1-example}.
To redistribute the multi-edges over these additional vertices, we replace in the gadget for the edge $e = v_{\alpha}^iv_{\beta}^j \in E(G_C)$ for every edge the vertices $u_{\alpha}^{\bot_L}$ with $u_{\alpha}^{\bot^e_L}$, $u_{\alpha}^{\bot_R}$ with $u_{\alpha}^{\bot^e_R}$, $u_{\beta}^{\bot_L}$ with $u_{\beta}^{\bot^e_L}$, and $u_{\beta}^{\bot_R}$ with $u_{\beta}^{\bot^e_R}$, respectively.
We observe that this eliminates all multi-edges created in the edge gadget, provided that $G_C$ has none, which we can assume.
In the fixation gadget, we only need to replace on page $p_e$ the edges $u_1^{\bot_L}u_1^{\bot}$ and $u_{k + 1}^{\bot}u_{k + 1}^{\bot_R}$ with $u_1^{\bot^e_L}u_1^{\bot}$ and $u_{k + 1}^{\bot}u_{k + 1}^{\bot^e_R}$, respectively.
This removes the remaining multi-edges from our construction.

What remains to do is to show that the properties of these gadgets still hold.
Once correctness of the gadgets has been established, correctness of the whole reduction follows immediately.
This can be seen by a closer analysis of the proof for \Cref{thm:kappa-w1}, together with the observation that the number of old vertices is still polynomial in the size of $G_C$ and the number of old edges and new vertices and edges is not affected.

\subsection{Correctness of Gadgets.} 
First, we can make the following observation for \Cref{prop:edge-gadget}, i.e., that from $u_{\alpha}^1 \prec x_{\alpha} \prec u_{\alpha}^{n_{\alpha} + 1}$, $u_{\beta}^1 \prec x_{\beta} \prec u_{\beta}^{n_{\beta} + 1}$, and $\sigma(x_{\alpha}x_\beta) = p_e$ follows that $x_{\alpha}$ is placed in  $\intervalPlacing{v_{\alpha}^i}$ and $x_{\beta}$ in $\intervalPlacing{v_{\beta}^j}$.
Since the prerequisites of \Cref{prop:edge-gadget} require $u_{\alpha}^1 \prec x_{\alpha} \prec u_{\alpha}^{n_{\alpha} + 1}$ and $u_{\beta}^1 \prec x_{\beta} \prec u_{\beta}^{n_{\beta} + 1}$ we (still) have $u_{\alpha}^{\bot^e_L} \prec u_{\alpha}^{\bot^e_R} \prec u_{\alpha}^1 \prec x_{\alpha} \prec u_{\alpha}^{n_{\alpha} + 1} \prec u_{\alpha + 1}^{\bot^e_L}$ and similarly for $\beta$, i.e., the relative position to the $\bot_L$- and $\bot_R$-vertices has not changed.
Thus, the same arguments as in the proof of \Cref{lem:edge-gadget-property} can be applied to our adapted reduction.

Regarding \Cref{prop:fixation-gadget-property}, i.e., that we have $u_{\alpha}^1 \prec x_{\alpha} \prec u_{\alpha}^{n_{\alpha} + 1}$, we observe that its proof relies, on the one hand, on our construction having \Cref{prop:bot-vertices-not-visible}, and, on the other hand, $x_{\alpha} \prec u_1^{\bot_L}$ being not possible due to $x_{\alpha}$ not being able to see $u_{\alpha + 1}^{\bot}$ on any page $p \in [\ell] \setminus \{p_d\}$.
For the latter claim, it is sufficient to see that for every page $p \in [\ell] \setminus \{p_d\}$ we still have that $x_{\alpha} \prec u_1^{\bot^p_L}$ is impossible since the edge $u_1^{\bot^p_L}u_1^{\bot}$ on page $p$ blocks visibility to $x_{\alpha + 1}^{\bot}$.
Using symmetric arguments that exclude $u_{\alpha}^{n_{\alpha}+1} \prec x_{\alpha}$, we can again conclude %
that we have %
$u_{\alpha}^1 \prec x_{\alpha} \prec u_{\alpha}^{n_{\alpha} + 1}$.

Combining all, we obtain the following observation.
\begin{observation}
	\label{obs:w-1-reduction-no-multi-edges-lemmas-preserved}
	If our reduction without multi-edges has \Cref{prop:bot-vertices-not-visible}, then we preserve the properties of the edge and fixation gadget, i.e., \Cref{prop:edge-gadget,prop:fixation-gadget-property} still hold.
\end{observation}

Until now, the correctness of the adapted reduction hinges on it having \Cref{prop:bot-vertices-not-visible}.
Thus, we now show that it still has said properties.
Recall \Cref{prop:bot-vertices-not-visible}, translated to our modified reduction:
\begin{property}
	\label{prop:bot-vertices-not-visible-no-multi-edges}
    For every color $\alpha \in [k]$ and %
    {%
    	for the page $p \in [\ell]$ of every edge gadget}, the vertex $u_{\alpha}^{\bot}$ is not visible on $p$ for any vertex $v$ with $u_1^{\bot^p_L} \prec v \prec u_{\alpha}^{\bot^p_L}$ or $u_{\alpha}^{n_{\alpha} + 1} \prec v \prec u_{k + 1}^{\bot^p_R}$.
\end{property}
We now establish the equivalent statement to \Cref{lem:bot-vertices-not-visible}.
\begin{lemma}
	\label{lem:bot-vertices-not-visibile-no-multi-edges}
	Our constructed instance of \QLE without multi-edges has \Cref{prop:bot-vertices-not-visible-no-multi-edges}.
\end{lemma}
\begin{proof}
	The arguments to show \Cref{lem:bot-vertices-not-visibile-no-multi-edges} are almost identical to those for showing \Cref{lem:bot-vertices-not-visible}.
	However, some statements in the different cases need slight adjustments, which we point out in the following.
	Recall that we consider a page $p_e \in [\ell] \setminus \{p_d\}$ for an edge $e = v_{\beta}^iv_{\gamma}^j \in E(G_C)$ and a placement of $v$ such that $u_1^{\bot^{p_e}_L} \prec v \prec u_{\alpha}^{\bot^{p_e}_L}$.
	Without loss of generality, we again assume $\beta < \gamma$.
	
	\proofsubparagraph*{Case 1: $v \prec u_{\beta}^i \prec u_{\gamma}^j$.}
        For the first sub-case, i.e., $\alpha \leq \beta$, we can now use the fact that we have $\sigma(u_1^{\bot^{p_e}_L}u_{\beta}^1) = p_e$ to conclude that $v$ cannot see $u_{\alpha}^{\bot}$.
	For the second sub-case, i.e., $\beta < \alpha$, we get the desired result by using $v \prec u_{\beta}^i \prec u_{\beta + 1}^{\bot^{p_e}_L} \prec u_{\alpha}^{\bot}$ and $\sigma(u_{\beta}^iu_{\beta + 1}^{\bot^{p_e}_L}) = p_e$.
		
	\proofsubparagraph*{Case 2: $u_{\beta}^i \prec v \prec u_{\gamma}^j$.}
	For this case there is nothing to do since we use in the proof of \Cref{lem:bot-vertices-not-visible} edges that do not have the $\bot_L$- and $\bot_R$-vertices as endpoints.
    Furthermore, since the relative placement of the new $\bot_L$- and $\bot_R$-vertices with respect to the other vertices remains unchanged, the arguments carry over without adaption.
	
	\proofsubparagraph*{Case 3: $u_{\beta}^i \prec u_{\gamma}^j \prec v$.}
	For the third and last case we get $u_{\gamma}^j \prec v \prec u_{\alpha}^{\bot} \prec u_{k + 1}^{\bot^{p_e}_R}$.
    Again, for $v \prec u_{\gamma}^{n_{\gamma} + 1}$ we can use the edge $u_{\gamma}^{n_{\gamma}+1}u_{\gamma + 1}^{\bot^{p_e}_L}$ to block visibility for $v$ and for $u_{\gamma}^{n_{\gamma} + 1} \prec v$ the edge $u_{\gamma}^{n_{\gamma}+1}u_{k + 1}^{\bot^{p_e}_R}$ is sufficient to conclude that $u_{\alpha}^{\bot}$ is not visible for $v$ on page $p_e$.

    As before, the setting with $u_{\alpha}^{n_{\alpha} + 1} \prec v \prec u_{k + 1}^{\bot^{p_e}_R}$ can be shown by an analogous, symmetric case analysis.
    Thus, \Cref{prop:bot-vertices-not-visible-no-multi-edges} holds in the construction without multi-edges.
\end{proof}

Finally, note that after (re-)establishing both properties, a closer look at \Cref{obs:fixation-gadget-edges} reveals that they only talk about edges incident to the vertices $u_{\alpha}^{\bot}$ and the page $p_d$.
Since we did not touch either of them, and have the corresponding \Cref{prop:bot-vertices-not-visible-no-multi-edges}, we conclude that \Cref{obs:fixation-gadget-edges} is also true for our modified reduction.
Similar holds for \Cref{obs:edge-gadget-colors}.

\end{document}